\documentclass{IEEEtran}

\usepackage{times}
\usepackage{amsmath}
\usepackage{amssymb}
\usepackage{amsthm}
\usepackage{color}
\usepackage{algorithm}
\usepackage[noend]{algorithmic}
\usepackage{graphicx}
\usepackage{subfig}
\usepackage{multirow}
\usepackage[bookmarks=false,colorlinks=false,pdfborder={0 0 0}]{hyperref}
\usepackage{cite}
\usepackage{bm}
\usepackage{arydshln}
\usepackage{mathtools}
\usepackage{microtype}

\newcommand{\sS}{\mathsf{S}}
\newcommand{\sT}{\mathsf{T}}
\newcommand{\sIn}{\mathsf{In}}
\newcommand{\sOut}{\mathsf{Out}}

\newtheorem{theorem}{Theorem}

\newtheorem{lemma}[theorem]{Lemma}

\long\def\symbolfootnote[#1]#2{\begingroup
\def\thefootnote{\fnsymbol{footnote}}\footnote[#1]{#2}\endgroup}

\title{Rack-Aware Regenerating Codes for Data Centers}

\author{Hanxu Hou, Patrick P. C. Lee$^{*}$, Kenneth W. Shum and Yuchong Hu
}

\begin{document}

\maketitle
\vspace{-0.5cm}
\begin{abstract}\symbolfootnote[0]{Hanxu Hou is with the School of Electrical Engineering \& Intelligentization, Dongguan University of Technology (E-mail: houhanxu@163.com). Patrick P. C. Lee is with the Department of Computer Science and Engineering, The Chinese University of Hong Kong (E-mail: pclee@cse.cuhk.edu.hk). Kenneth W. Shum is with the School of Science and Engineering, The
Chinese University of Hong Kong (Shenzhen). This work was done when he
was with Institute of Network Coding, The Chinese University of Hong
Kong. (E-mail: wkshum@inc.cuhk.edu.hk). Yuchong Hu is with the School of Computer Science and Technology, Huazhong University of
Science and Technology (E-mail: yuchonghu@hust.edu.cn). This work was partially supported by the
National Natural Science Foundation of China (No. 61701115,61872414,61502191) and Research Grants Council of Hong Kong (GRF 14216316 and CRF C7036-15G) and The Chinese University of Hong Kong - Shanghai Jiao Tong University
Joint Research Collaboration Fund (No. 4750358) and Fundamental Research Funds for the Central Universities (No. 2017KFYXJJ065,2016YXMS085), Alibaba Innovation Research.\\
* Corresponding author.}
Erasure coding is widely used for massive storage in data centers to
achieve high fault tolerance and low storage redundancy.  Since the cross-rack
communication cost is often high, it is critical to design erasure codes that
minimize the cross-rack repair bandwidth during failure repair.  In this
paper, we analyze the optimal trade-off between storage redundancy and
cross-rack repair bandwidth specifically for data centers, subject to the
condition that the original data can be reconstructed from a sufficient number
of any non-failed nodes.  We characterize the optimal trade-off curve under
functional repair, and propose a general family of erasure codes called
{\em rack-aware regenerating codes (RRC)}, which achieve the optimal trade-off.
We further propose exact repair constructions of RRC that have minimum storage
redundancy and minimum cross-rack repair bandwidth, respectively. We show that
(i) the minimum storage redundancy constructions support a wide range of
parameters and have cross-rack repair bandwidth that is strictly less than
that of the classical minimum storage regenerating codes in most cases, and
(ii) the minimum cross-rack repair bandwidth constructions support all the
parameters and have less cross-rack repair bandwidth than that of the minimum
bandwidth regenerating codes for almost all of the parameters.
\end{abstract}


\IEEEpeerreviewmaketitle

\section{Introduction}

Modern storage systems are often deployed in the form of data centers, in
which data is distributed across a large number of storage nodes that are
grouped in different racks.  Examples include Google File System
\cite{ford2010} and Windows Azure Storage~\cite{huang2012}, and Facebook
storage \cite{sathiamoorthy2013}.  To provide high availability and durability
for data storage against node failures, erasure coding is now widely adopted
in modern storage systems to encode data with significantly higher fault
tolerance and lower storage redundancy in compare to traditional replication.
In particular, Reed-Solomon (RS) codes \cite{reed1960} are the most popular
erasure codes that are adopted in production (e.g., in Google~\cite{ford2010}).
An $(n,k)$ RS code encodes a data file of $k$
{\em symbols} (i.e., the units for erasure coding operations) to obtain $n$
symbols over some finite field, and distributes the $n$ symbols in $n$
different nodes (where $k < n$).  The data file can then be retrieved by a data collector by
connecting to any $k$ out of $n$ nodes via a \emph{reconstruction} process.  RS
codes have two important practical advantages: (i) they achieve minimum
storage redundancy while tolerating any $n-k$ node (or symbol) failures, and
(ii) they support arbitrary values of $n$ and $k$ ($< n$).

When a node fails, each lost symbol stored in the failed node needs to be
repaired in a new node to maintain the same level of fault tolerance.  The
conventional repair method, which is also used by RS codes, is to first
reconstruct the data file and then encode it again to form each lost symbol.
Thus, for an $(n,k)$ RS code, the total amount of data downloaded to repair a
lost symbol is $k$ symbols (i.e., $k$ times of the lost data).  This amplifies
both network bandwidth and I/O cost.


The concept of \emph{regenerating codes (RC)} is formulated by Dimakis
{\em et al.} \cite{dimakis2010} with the objective of minimizing the network
bandwidth during a repair operation.  RC encodes a data file into a multiple
of $n$ symbols and distributes them to $n$ nodes, each of which stores
multiple symbols, while the data file can still be reconstructed from a
sufficient number of nodes as in RS codes.  To repair the lost symbols of a
failed node in a new node, the new node retrieves {\em encoded} symbols from
each of a selected subset of non-failed nodes, where the encoded symbols are
derived from the stored symbols.  In general, the total amount of encoded
symbols retrieved from all non-failed nodes, also known as the {\em repair
bandwidth}, is much less than the original data file size.  Dimakis
{\em et al.} \cite{dimakis2010} also characterize the optimal trade-off
between repair bandwidth and storage redundancy.

In practical data centers, storage nodes are organized in racks, and the
cross-rack communication cost is typically much more expensive than the
intra-rack communication cost.  It is thus important for erasure codes to
specifically minimize the {\em cross-rack repair bandwidth} (i.e., the total
amount of symbols transferred across different racks during a repair process).
Unfortunately, RC does not address this constraint, and
generally cannot minimize the cross-rack repair bandwidth.  This motivates
a number of studies that specifically address the repair problem for data
centers (see Section~\ref{sec:related} for details).  In particular, Hu
{\em et al.} \cite{hu2016} propose double regenerating codes (DRC) to
minimize the cross-rack repair bandwidth by reconstructing partially
repaired symbols locally within each rack and combining the partially repaired
symbols across racks.  It is shown that DRC can achieve much less cross-rack
repair bandwidth than RC for some choices of code parameters.  However, DRC is built on
the condition that the minimum storage redundancy is achieved (as in RS
codes).  The optimal trade-off between storage redundancy and cross-rack
repair bandwidth, similar to the optimality analysis for RC
\cite{dimakis2010}, remains largely unexplored in the context of erasure-coded
data centers.


\subsection{Contributions}
In this paper, we consider a more general model of DRC \cite{hu2016}, in the
sense that the model supports a flexible storage size and a flexible number of
non-failed nodes that provide repaired data during a repair process.  To this
end, we propose a general family of erasure codes called {\em Rack-aware
Regenerating Codes (RRC)} for data centers.  The main contributions of this
paper are as follows.
\begin{itemize}
\item First, we derive the trade-off between storage and cross-rack repair
bandwidth of RRC. In the optimal trade-off curve, there exist two extreme points,
namely {\em minimum storage rack-aware regeneration (MSRR)} and {\em minimum
bandwidth rack-aware regeneration (MBRR)} points, which correspond to the
minimum storage and the minimum cross-rack repair bandwidth, respectively.
The trade-off curve of RRC can be reduced to the optimal trade-off curve of
RC if each rack has one node.
Let $r$ be the number of racks. When $kr/n$ is an integer,
the trade-off of MSRR codes is exactly the same as that of the minimum storage
regeneration (MSR) codes; when $kr/n$ is not an integer, the cross-rack repair
bandwidth of MSRR codes is strictly less than that of MSR codes. Also, we show
that the cross-rack repair bandwidth of our MBRR codes is strictly less than
that of the minimum bandwidth regeneration (MBR) codes for most of the parameters
(see Theorem~\ref{thm:cmpr} for details). For example, when $(n,k,r)=(12,8,4)$,
MSRR codes have 33.3\% reduction of cross-rack repair bandwidth compared to
MSR codes; for the same parameters, MBRR codes achieve 13.1\% and 28.9\%
reduction of cross-rack repair bandwidth and storage over MBR codes,
respectively (see Fig.~\ref{storage-bandwidth} for details).
Compared to the related work, the cross-rack repair bandwidth
of our RRC is less than or equal to that of the codes in \cite{prakash2017}
for all parameters, and less than that of the codes in \cite{sohn2016} for
most parameters (see Section~\ref{sec:com-repair} for details).
\item
Second, we present several constructions for MSRR codes with exact repair, which support a much wider range of parameters than those in \cite{hu2016,prakash2017,sohn2016}. We also present an exact-repair construction for the MBRR codes, which support all the parameters, again an improvement compared to \cite{hu2016,prakash2017} (see Section~\ref{sec:parameter} for
details).
Note that the exact-repair construction of the codes in
\cite{sohn2016} is given in \cite{sohn2018a,sohn2018}. For example, when
$n=12$ and $r=4$, our MSRR code construction can support $k=4,5,\ldots,11$,
while the constructions in \cite{hu2016} and \cite{prakash2017} can only
support $k=9$ and $k=6,9$, respectively. Note that the exact-repair
construction of the minimum storage codes in \cite{sohn2016} is given in the
later work \cite{sohn2018a}, and it can only support $r=2$ and $n=2k$.
\end{itemize}


\subsection{Paper Organization}

The rest of the paper is organized as follows.
Section~\ref{sec:related} reviews the related work.
Section~\ref{sec:system_model} introduces the system model.
Section~\ref{sec:tradeoff} shows the optimal trade-off between storage and
cross-rack repair bandwidth.
Section~\ref{sec:cons} gives the exact-repair constructions for MSRR codes.
Section~\ref{sec:MBRRcons} gives the exact-repair construction of MBRR codes
for all parameters.
Section~\ref{sec:comp} presents evaluation results for our RRC and the related
codes.
Section~\ref{sec:conclu} concludes the paper.

\section{Related Works}
\label{sec:related}

There are many follow-up studies on RC along different directions, such as
practical implementation \cite{chen2014,rashmi2015,pamies2016,li2017} and the
repair problem with heterogeneous structures
\cite{pernas2013,ernvall2013,tebbi2014,hu2016,sohn2016,hu2017,prakash2017}.

Flexible RC \cite{shah2010a} is designed for heterogeneous storage systems that
can achieve the lower bound of repair bandwidth.
Combined with a tree-structured regeneration topology, it is shown that RC can
further save the network bandwidth \cite{li2010,wang2014}. Some studies
\cite{akhlaghi2010,ernvall2013} focus on the capacity bound for a
heterogeneous model.  However, all the above studies do not distinguish the
costs between intra-rack and cross-rack communications in data centers.

\begin{table*}
\caption{Comparison with related work.}
\begin{center}
\vspace{-9pt}
\begin{tabular}{|c|c|c|c|c|}
\hline
 &$(n,k)$ recovery    &  Cross-rack repair bandwidth & Supported parameters & Supported parameters \tabularnewline
 &property & & of minimum storage & of minimum bandwidth\\
\hline
\hline
RC & holds &   $\geq$ RRC, equality holds for MSR  & all parameters \cite{goparaju2017} & all parameters  \tabularnewline
 &  &    only when $kr/n$ is an integer   &  &   \tabularnewline
\hline
DRC \cite{hu2016,hu2017} & holds & $=$ MSRR & $\frac{n}{n-k}$ is an integer or $r=3$ & n/a \tabularnewline
 \hline
Sohn {\em et al.} \cite{sohn2016} & holds &  $\geq$ RRC, equality holds & $n=2k,r=2$ \cite{sohn2018a}  & all parameters \cite{sohn2018}  \tabularnewline
 &  &   when $kr/n$ is an integer &  &   \tabularnewline
\hline
Prakash {\em et al.} \cite{prakash2017} & does not hold &  $=$ RRC for $kr/n$ is an integer & $kr/n$ is an integer & $kr/n$ is an integer \tabularnewline
\hline
\emph{RRC} (this paper) & holds & $\leq$ \cite{sohn2016,prakash2017}, equality holds & most parameters & all parameters \tabularnewline
 &  & when $kr/n$ is an integer &   &   \tabularnewline
\hline
\end{tabular}
\end{center}
\label{tb:cmp}
\end{table*}

Some previous studies distinguish the costs between cross-rack and intra-rack
communications, yet their system models and analysis are fundamentally
different from ours.  Table~\ref{tb:cmp} compares our RRC with several
closely related work for erasure-coded data centers.
DRC \cite{hu2016,hu2017} considers
the same model of this paper and achieves the trade-off between storage and
cross-rack repair bandwidth under the minimum storage condition. DRC can be
viewed as a special case of our MSRR codes with all other racks being
contacted to repair a failed node.
Sohn {\em et al.} \cite{sohn2016} consider a different repair
model and give the optimal trade-off between storage and repair bandwidth
(including cross-rack repair bandwidth and intra-rack repair bandwidth).  In
their repair process, there is no information encoding between two nodes in
the same rack, while in our model, the symbols downloaded from other racks are
the combinations of all the symbols in the rack (as in DRC
\cite{hu2016,hu2017}).  Also, to repair a failed node, the new node in
\cite{sohn2016} needs to connect to all the other racks, while the number of
racks connected to repair a failed node is more flexible in our paper. We can
show that the cross-rack repair bandwidth of RRC is less than that of the
codes in \cite{sohn2016} for most parameters.
Later, Sohn {\em et al.} \cite{sohn2018a,sohn2018} present
exact-repair constructions for the minimum storage point and the minimum
bandwidth point of the codes in \cite{sohn2016}.

The closest related work to ours is by Prakash {\em et al.}
\cite{prakash2017}.  In their model, a file needs to be retrieved from a
certain number of racks, and hence $k$ must be a multiple of the number of
nodes in each rack.  On the other hand, our model allows a file to be
retrieved from any $k$ nodes.  Therefore, our RRC can tolerate more failure
patterns than the codes in \cite{prakash2017}. We show that the trade-off
curve of RRC coincides with the optimal trade-off curve in \cite{prakash2017}
when $k$ is a multiple of the number of nodes in each rack, yet
our exact-repair constructions for MSRR codes and MBRR codes can support much
more parameters than that of the minimum storage codes and the minimum
bandwidth codes in \cite{prakash2017}, respectively (see
Section~\ref{sec:parameter} for details). More importantly, the cross-rack
repair bandwidth of our MSRR codes with additional parameters is strictly less
than that of MSRR codes with the nearest $k$ that is a multiple of the number
of nodes in each rack (see the remark in Section~\ref{sec:tradeoff}).
In other words, the minimum storage codes and the minimum
bandwidth codes in \cite{prakash2017} only support the parameters when $k$ is
a multiple of the number of nodes in each rack, while our MSRR codes and MBRR
codes do not have this restriction on $k$.  Note that when $k$ is a multiple
of the number of nodes in each rack, $k+1$ will not be a multiple of number of
nodes in each rack. We can show that the cross-rack repair bandwidth of MSRR
(resp.  MBRR) codes with $k+1$ data nodes is strictly less than that of MSRR
(resp. MBRR) codes with $k$ data nodes.

There are other studies that specifically address the deployment of erasure
coding in rack-based data centers.  Some studies \cite{gaston2013,pernas2013}
consider the trade-off with two racks.  Tebbi {\em et al.}~\cite{tebbi2014}
design locally repairable codes for multi-rack storage systems.  Shen {\em et
al.} \cite{shen2016} present a rack-aware recovery algorithm that is
specifically designed for RS codes.  In this paper, we conduct formal analysis
and formulate a general model that gives the optimal trade-off between storage
and cross-rack repair bandwidth.

A similar methodology in the two-layer coding for data centers can be found in \cite{gerami2017}. The first layer encodes the data file by an
$(n,k)$ MDS code and distributes to $n$ nodes, while the second layer creates
the symbols stored in each node by employing an MDS code with the code rate
$\delta$.  If the proportion of the failed symbols among the symbols stored in
a node is no larger than $1-\delta$ (i.e., a partial node failure), then the
failed symbols can be recovered by the node locally. Otherwise, there is a
trade-off between storage and repair bandwidth. The main difference between
the work in \cite{gerami2017} and our work is that we distinguish the
intra-rack and cross-rack communications and consider the repair of a failed
node in a rack-based storage system, while the authors in \cite{gerami2017}
consider partial node failures.

\section{System Model}
\label{sec:system_model}

\begin{figure}[t]
\centering
\includegraphics[width=0.5\textwidth]{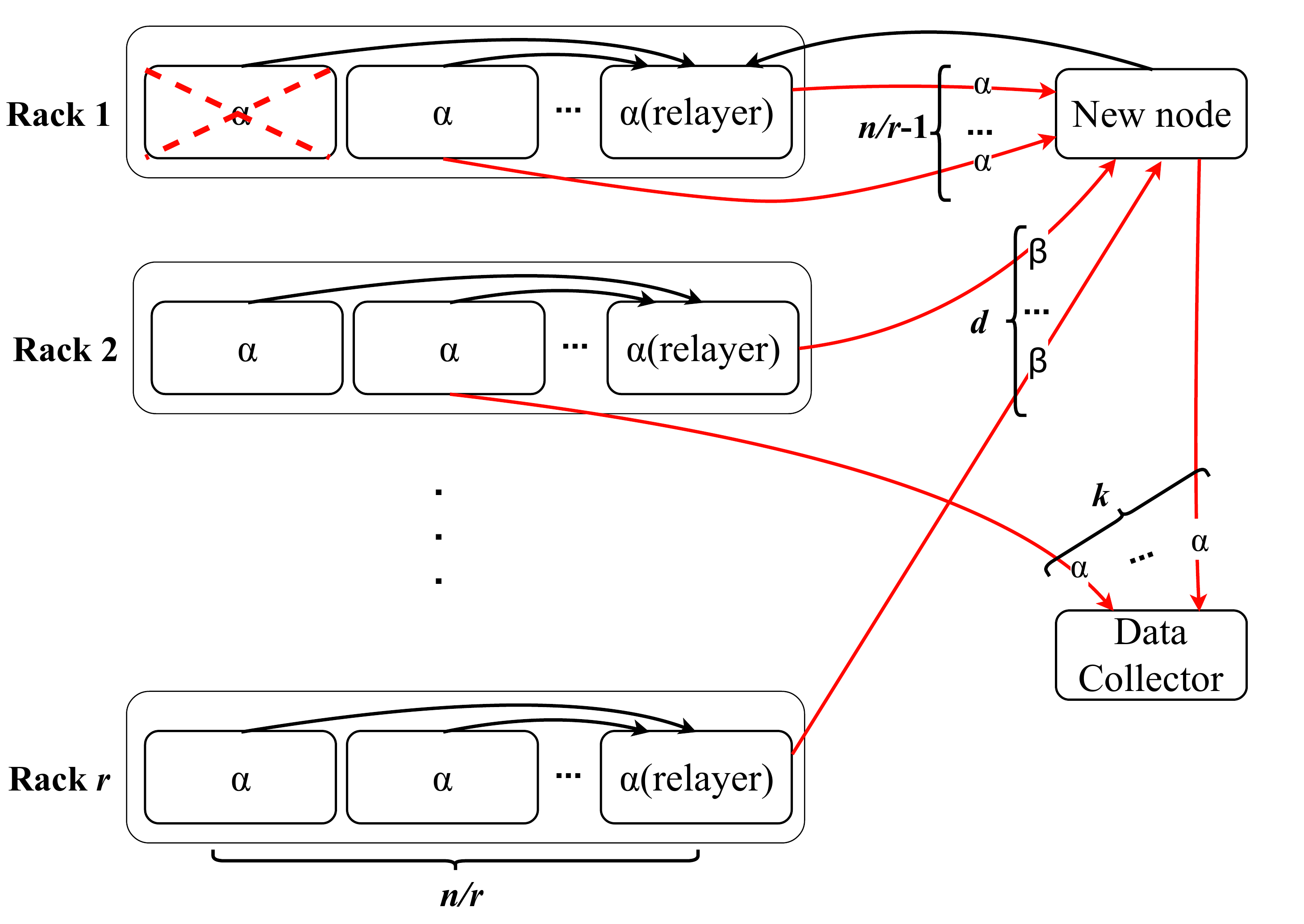}
\vspace{-9pt}
\caption{A failed node can be repaired by downloading all the other symbols in
the host rack and $\beta$ symbols each from $d$ other arbitrary racks. The
data file can be reconstructed by a data collector by downloading $k\alpha$
symbols from any $k$ nodes.}
\label{system-model}
\end{figure}

We consider a data center consisting of $n$ nodes that are equally divided
into $r$ racks, with $n/r$ nodes in each rack (see Fig.~\ref{system-model}).
We assume throughout this paper that  $n$ is a multiple of $r$, and label the
nodes from $1$ to $n$. For $h=1,2,\ldots, r$ and $i=1,2,\ldots, n/r$, we
denote the $i$-th node in rack $h$  by $X_{h,i}$. We fix an alphabet of size
$q$. A {\em data file} is regarded as a sequence of $B$ symbols. A data file
is encoded into $n\alpha$ symbols and stored in $n$ nodes. Each node stores
$\alpha$ symbols.

In each rack, we select a distinguished node called the \emph{relayer node}
for each data file, such that the relayer node can obtain the content stored
in the other nodes in the same rack.
We assume that the intra-rack bandwidth is abundant, so that
the transmissions among the nodes within a rack incurs negligible cost. If a
storage node fails, we replace it by a new node and put it in the same rack.
The new node arbitrarily picks $d$ other racks, where $d < r$, and
connects to the corresponding relayers. We call the relayers or racks that
participate in the repair process to be {\em helpers}, and the parameter $d$
to be the {\em repair degree}.
Based on the $\alpha n/r$ symbols stored in the host rack, each of the
contacted relayers sends $\beta$ symbols to the new node. The cross-rack repair
bandwidth is $\gamma=d\beta$. The content of the new node is then regenerated
from the received $d\beta$ symbols and the $(n/r-1)\alpha$ symbols stored in
the host rack.
Note that the relayer can be any arbitrary surviving node
selected from a rack, and different data files can be associated with different
relayers during a repair operation.  We can view a node failure as a partial
failure of a rack.  We can repair a failed node by downloading $\beta$ symbols
from each of any other $d$ racks, and $(n/r-1)\alpha$ symbols from the other
$n/r-1$ nodes in the same rack. By relabeling the storage node, we assume
that $X_{h,1}$ is the relayer in rack $h$, for $h=1,2,\ldots, r$, without loss
of generality.



\begin{table}[t]
\vspace{-0.2cm}
\caption{Main notation used in this paper.}
\vspace{-9pt}
\begin{center}
\begin{tabular}{|c|c|}
\hline
Notation & Description \tabularnewline
\hline
\hline
$n$ & number of nodes \tabularnewline
\hline
$r$ & number of racks \tabularnewline
\hline
$B$ & number of symbols in a data file \tabularnewline
\hline
$\alpha$ & number of symbols stored in each node \tabularnewline
\hline
$n/r$ & number of nodes in each rack \tabularnewline
\hline
$d$ & repair degree \\
\hline
$\beta$ & number of symbols downloaded from a relayer \\
\hline
$\gamma=d\beta$ & cross-rack repair bandwidth \tabularnewline
\hline
$G(n,k,r,d,\alpha,\beta)$ & information flow graph \tabularnewline
\hline
$\alpha^{*}(\beta)$ & minimum $\alpha$ for a given $\beta$ \tabularnewline
\hline
Defined in Section~\ref{sec:tradeoff} & \\
\hline
$m$ & the value of $\lfloor kr/n \rfloor$ \tabularnewline
\hline
$t$ & the value of $k \bmod (n/r)$ \tabularnewline
\hline
Defined in Section~\ref{sec:cons} & \\
\hline
$\mathbf{s}$ & $B$ data symbols \tabularnewline
\hline
$\mathbf{Q}_{i,h}$ & encoding matrix of node $i$ in rack $h$ \tabularnewline
\hline
$\mathbf{G}_{i}$ & encoding matrix of rack $i$ \tabularnewline
\hline
$\mathbf{c}^{T}_{f,h}$ & local encoding column of rack $h$ \\
 & to repair a node in rack $f$\\
\hline
$0_{\alpha}$ & $\alpha\times\alpha$ zero matrix \tabularnewline
\hline
$I_{\alpha}$ & $\alpha\times\alpha$ identity matrix \tabularnewline
\hline
$\mathbf{P}_i$ & the $B\times m$ left-most sub-matrix of $\mathbf{G}_i$\\
\hline
$\mathbf{R}_i$ & the $B\times (n-k-m)$ right-most\\
 &  sub-matrix of $\mathbf{G}_i$\\
\hline
$\mathbf{v}_1,\ldots,\mathbf{v}_{m}$ & $m$ orthogonal row vectors of length $m$\\
\hline
$\mathbf{u}_1,\ldots,\mathbf{u}_{r-m}$ & $m$ row vectors of length $n-k$\\
\hline
$\mathbf{G}[(i_1,i_2),(j_1,j_2)]$ & sub-matrix of $\mathbf{G}$ consisting from \\
  & rows $i_1$ to $i_2$ and columns $j_1$ to $j_2$\\
\hline
$\mathbf{E}_i$ & a $(n-k)\times m$ random matrix over $\mathbb{F}_q$\\
\hline
$\lambda_{i,j}$ & a non-zero element of $\mathbb{F}_q$\\
\hline
$\mathbf{y}_1,\ldots,\mathbf{y}_{m}$ & $m$ orthogonal row vectors of length $\alpha n/r-m$\\
\hline
$\mathbf{x}_2,\ldots,\mathbf{x}_{m}$ & $\alpha-1$ orthogonal row vectors of length $\alpha n/r$\\
\hline
\end{tabular}
\end{center}
\label{tb:notat}
\vspace{-0.3cm}
\end{table}

We want to maintain the property that any $k$ nodes suffice to decode
the data file. We call this the $(n,k)$ {\em recovery property}.
When a data collector connects to a relayer node, it is equivalent to
connecting to all the $n/r$ nodes in the rack. Without loss of generality, we can
make the assumption that if a data collector connects to a relayer, it also
connects to all of the other nodes in the same rack.
We consider two versions of repair
in this paper: exact repair and functional repair. In
{\em exact repair}, the symbols stored in the failed node are the same as
those in the new node. In {\em functional repair}, the new node may contain
symbols different from those in the failed node, as long as the $(n,k)$ recovery
property is preserved.  An encoding scheme that satisfies all of the above
requirement with parameters $n$, $k$, $r$, $d$, $\alpha$ and $\beta$ is called
a {\em rack-based storage system} $RSS(n,k,r,d,\alpha,\beta)$.
Table~\ref{tb:notat} summarizes the main notation used in this paper.

\section{Optimal Trade-off Between Storage and Cross-rack Repair Bandwidth}
\label{sec:tradeoff}

We represent the storage system described in the previous section by an {\em
information flow graph}, which was proposed in \cite{dimakis2010} for deriving
optimal trade-off of RC. In order to differentiate from the system diagram
in Fig.~\ref{system-model}, we will use the term ``vertex'', instead of
``node'', for the information flow graph.


Given the system parameters $n$, $k$, $r$, $d$, $\alpha$, and $\beta$, an
information flow graph is a directed acyclic graph (DAG) constructed according
to the following rules. There is a vertex $\sS$ that represents the data file,
and a vertex $\sT$ that represents the data collector. For $h=1,2,\ldots, r$
and $i=1,2,\ldots, n/r$, the $i$-th node in rack $h$ is represented by a pair
of vertices $\sIn_{h,i}$ and $\sOut_{h,i}$. We draw an edge from $\sIn_{h,i}$
to $\sOut_{h,i}$ with capacity $\alpha$.  To each in-vertex $\sIn_{h,i}$,
we draw an edge from $\sS$ to $\sIn_{h,i}$ with infinite capacity. This
represents the encoding process as the content in each storage node is a
function of all the symbols, and the capacity of each node is limited to
$\alpha$.  For each $h=1,2,\ldots, r$ and $i=2,3,\ldots, n/r$, we draw an edge
with infinite capacity from $\sOut_{h,i}$ to $\sOut_{h,1}$. This indicates
that $X_{h,1}$ is the relayer, and $X_{h,1}$ can access everything stored in
$X_{h,i}$.

Suppose that the $f$-th node in rack $h$ fails, for some $h\in\{1,2,\ldots, r\}$
and $f\in\{1,2,\ldots, n/r\}$.  We put $n/r$ pairs of vertices, say
$\sIn_{h,j}'$ and $\sOut_{h,j}'$ in the information flow graph. For
$j\in\{1,2,\ldots, n/r\}\setminus\{f\}$, we draw an edge with infinite
capacity from $\sOut_{h,j}$ to $\sIn_{h,j}'$, and an edge with infinite
capacity from $\sIn_{h,j}'$ to $\sOut_{h,j}'$. This means that the content of
node $j$ does not change after the repair. For vertex $\sIn_{h,f}'$, which
represents the new node, we draw an edge from $\sOut_{h,j}$ to $\sIn_{h,f}'$
with infinite capacity, indicating that it can access all the symbols stored
in the other nodes in the same rack. Suppose that the new node makes $d$
connections to the relayers in rack $h_1, h_2,\ldots, h_d$, where $h_1,\ldots,
h_d$ are distinct indices that are not equal to~$h$. There is an edge with
capacity $\beta$ in the information flow graph from
$\sOut_{h_\ell,1}$ to $\sIn_{h,f}'$, for $\ell=1,2,\ldots, d$.
Thus, $\sIn_{h,f}'$ has $(n/r-1)+d$ incoming edges, in which $d$ of them
have capacity $\beta$ and $n/r-1$ of them have infinite capacity.
The new node stores $\alpha$ symbols eventually, and we represent
this by drawing an edge from $\sIn_{h,f}'$ to $\sOut_{h,f}'$ with
capacity $\alpha$. We also have an edge with infinite capacity
from $\sOut_{h,j}'$ to $\sOut_{h,1}'$ for $j=2,3,\ldots, n/r$.

\begin{figure}
\centering
\includegraphics[width=0.9\linewidth]{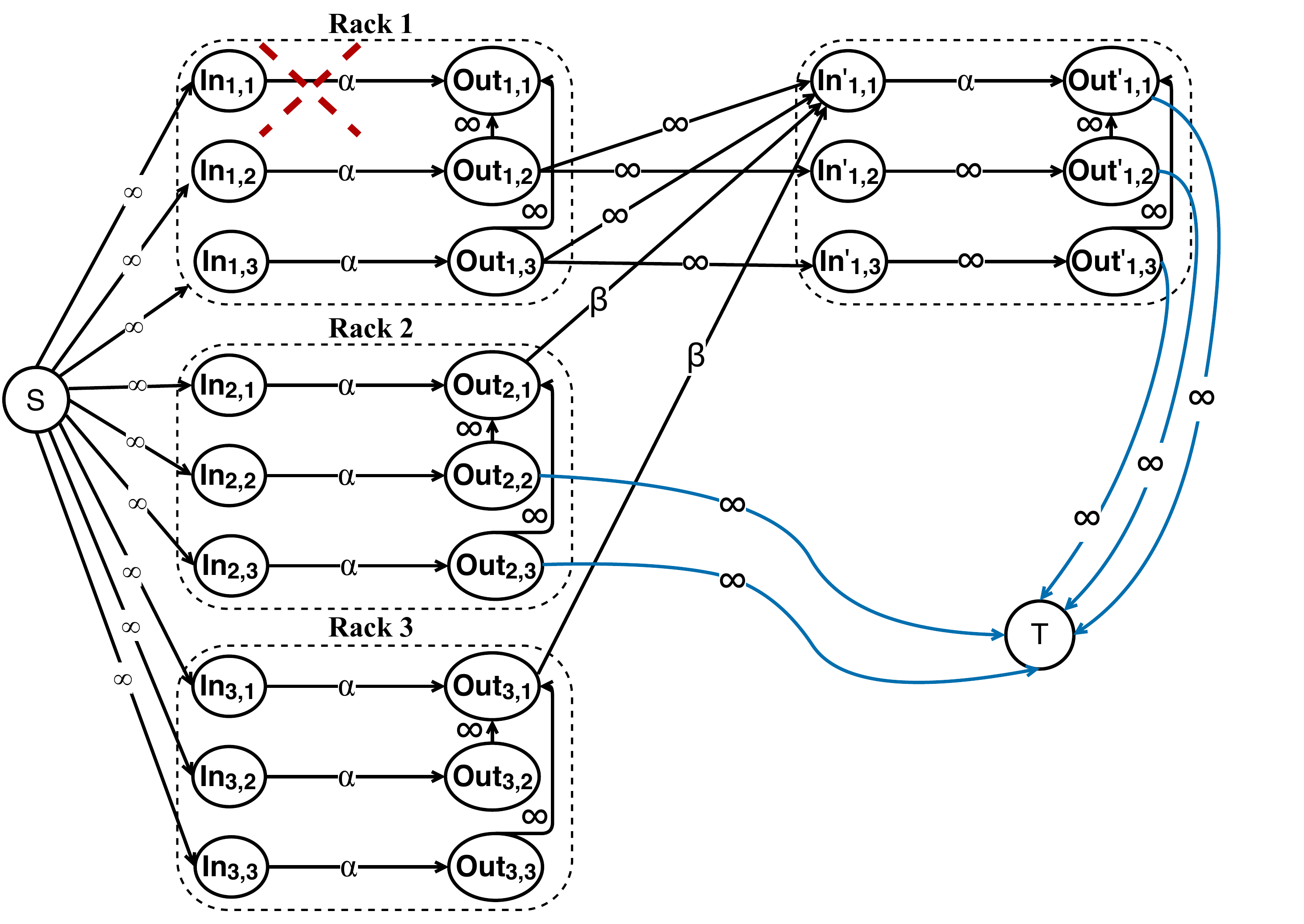}
\caption{Information flow graph of $(n,k,r,d)=(9,5,3,2)$.}
\label{model-example}
\vspace{-6pt}
\end{figure}


The storage system may undergo a series of node failures and repairs. We
repeat the above procedure accordingly.  Finally, we draw $k$ edges from $k$
out-vertices to  $\sT$. We keep the convention that if $\sT$ is connected to
the vertex $\sOut_{h,1}$ corresponding to the relayer in rack $h$, $\sT$ is
also connected to all the vertices $\sOut_{h,2},\ldots, \sOut_{h,n/r}$ in rack
$h$.

Any DAG that can be obtained as described above is referred
to as an information flow graph, and is denoted by $G(n,k,r,d,\alpha, \beta)$.
Fig.~\ref{model-example} shows an example of $(n,k,r,d)=(9,5,3,2)$.

Given an information flow graph $G$, we regard the unique vertex $\sS$ as the
source vertex and the unique vertex $\sT$ as the terminal vertex, and consider
the maximum flow from $\sS$ to $\sT$. We define an $(\sS,\sT)$-{\em cut} as a
subset of the edges in $G$ such that $\sS$ and $\sT$ are disconnected after
the edges in this subset are removed from~$G$. The {\em capacity} of an
$(\sS,\sT)$-cut is defined as the sum of the capacity of the edges in the cut.
Let $\text{mincut}(G)$
denote the smallest capacity of an $(\sS,\sT)$-cut in a given information flow  graph $G$, and
$\min_G \text{mincut}(G)$
with the minimum value taken over all possible information flow graphs $G$. By
the max-flow bound in network coding theory \cite[Theorem 18.3]{yeung08}, the
supported file size $B$ cannot exceed $\min_G \text{mincut}(G)$. The next
theorem determines $\min_G \text{mincut}(G)$, and hence gives an upper bound
on the file size.  Throughout the paper, we will use the notation $$m:=\lfloor
\frac{kr}{n}\rfloor.$$




\begin{theorem}
Given the parameters $n$, $k$, $r$, $d\geq m$, $\alpha$, $\beta$ and $B$, if there is an $RSS(n,k,r,d,\alpha,\beta)$ with file size $B$, then
\begin{equation}
k\alpha+\sum_{\ell=1}^{m}\min\{(d-\ell+1)\beta-\alpha,0\}\geq B.
\label{eq1}
\end{equation}
\label{thm1}
\vspace{-0.3cm}
\end{theorem}
The proof of Theorem~\ref{thm1} is given in Appendix A.



If an encoding scheme for $RSS(n,k,r,d,\alpha,\beta)$ with the equality in~\eqref{eq1} holds, we call it a {\em rack-aware regenerating code}
$RRC(n,k,r,d,\alpha,\beta)$. The value on the left-hand side of the inequality
in~\eqref{eq1} is called the {\em capacity} of $RRC(n,k,r,d,\alpha,\beta)$.
When $r=n$, we note that the trade-off curve of RRC in~\eqref{eq1} reduces to
the optimal trade-off curve of RC \cite{dimakis2010}.

\textbf{Remark.} If $kr/n$ is an integer (i.e., $m=kr/n$),
then the upper bound given in~\eqref{eq1} is the same upper bound obtained
from \cite{prakash2017} (see (1) in \cite{prakash2017}).  Note that the repair
scenarios of our work and \cite{prakash2017} are the same, yet our model can
tolerate more failure patterns than the model in \cite{prakash2017}. If $kr/n$ is not an integer, our upper bound is tighter than that of
the bound given in \cite{prakash2017}.

We now characterize the achievable trade-offs between the storage $\alpha$ and
the cross-rack repair bandwidth $\gamma=d\beta$ for given $(n,k,r,d)$. Given
$\beta$, $\alpha^{*}(\beta)$ is defined to be the smallest $\alpha$ such that
the equality in~\eqref{eq1} holds if such a solution exists, and is set to be
infinity otherwise.  The following theorem shows the optimal trade-off.

\begin{theorem}
\label{thm:opt-tradeoff}
Given the parameters $n$, $k$, $r$, $d$ and $B$, we let
\[
g(i)=i\frac{2d-2m+i+1}{2d},
\]
\[
f(i)=\frac{2B}{2k(d-m+1)+i(2k-i-1)},
\]
where $i=0,1,\ldots,m-1$. If $\beta$ ranges from $f(m-1)$ to infinity, then the minimum storage $\alpha^{*}(\beta)$ is as follows,
\begin{equation}
\alpha^{*}(\beta)=\left\{
\begin{array}{ll}
\frac{B}{k}, &\text{ } \beta \in [f(0),+\infty)\\
\frac{B-g(\ell)d\beta}{k-\ell},& \text{ } \beta\in[f(\ell),f(\ell-1)),
\end{array}
\right.
\label{eq:min-storage1}
\end{equation}
for $\ell=1,2,\ldots,m-1$, and
\begin{equation}
\alpha^{*}(\beta)= \frac{Bd}{(k-m)d+m(d-\frac{m-1}{2})},
\label{eq:min-storage2}
\end{equation}
 for $\beta=
f(m-1)$.
\end{theorem}
\begin{proof}
See Appendix B.
\end{proof}

%


There are two extreme points on the optimal trade-off curve that correspond
to the minimum storage and the minimum cross-rack repair bandwidth.
The two extreme points are called {\em minimum storage rack-aware regenerating
(MSRR)} codes and {\em minimum bandwidth rack-aware regenerating (MBRR)}
codes, respectively.  The MSRR point can be derived by first minimizing
$\alpha$ and then $\beta$, while the MBRR point can be derived by first
minimizing $\beta$ and then $\alpha$.

From~\eqref{eq:min-storage1} and~\eqref{eq:min-storage2}, the MSRR point can be achieved when
\begin{equation}
(\alpha_{\mathrm{MSRR}},\gamma_{\mathrm{MSRR}})=(\frac{B}{k},\frac{Bd}{k(d-m+1)}),
\label{eq:mscr}
\end{equation}
and the MBRR point is achieved by
\begin{equation}
\begin{array}{c}
\alpha_{\mathrm{MBRR}}=\gamma_{\mathrm{MBRR}}=\frac{Bd}{(k-m)d+m(d-\frac{m-1}{2})}.
\end{array}
\label{eq5}
\end{equation}

Observe that when $r=n$,
MSRR codes are reduced to minimum storage regenerating (MSR) codes, and MBRR
codes are reduced to minimum bandwidth regenerating (MBR) codes. Note that in
MBRR codes, the cross-rack repair bandwidth $\gamma$ is equal to the
storage $\alpha$, i.e., the amount of data downloaded from other racks has the
same size as the failure data.
In MBRR codes, we have $B=kd-m(m-1)/2$ and $\alpha=\gamma=d$
according to~\eqref{eq5}. If $2kd-m(m-1)>nd$, then the code rate (i.e.,
$\frac{B}{n\alpha}$) of MBRR codes is
\[
\frac{kd-m(m-1)/2}{nd}>0.5.
\]

Therefore, we may construct MBRR codes with high code rates, while the code
rates of all MBR codes are no larger than $0.5$. Given $r$ and $n$, the values
of $k$ for which the code rates of MBRR codes are larger than $0.5$ are
summarized in Table~\ref{tb:mbrr} for $d=r-1$. For all the evaluated
parameters in Table~\ref{tb:mbrr}, the MBRR codes have high code rates when
$k/n>0.5$.

\begin{table}
\caption{Some parameters of $r,n,k$ for which the MBRR codes have high code
rates.}
\begin{center}
\vspace{-9pt}
\begin{tabular}{|c|c|c|c|c|c|}
\hline
$(r,n)$ &$(4,8)$    &  $(4,12)$ &  $(4,16)$ & $(5,10)$ & $(5,15)$ \tabularnewline
\hline
 $k$ & $5-7$ & $7-11$  & $9-15$   & $6-9$ & $8-14$  \tabularnewline
\hline
\hline
$(r,n)$  & $(5,20)$ & $(6,12)$ & $(6,18)$ & $(6,24)$ & $(6,30)$\tabularnewline
\hline
 $k$ & $11-19$ & $7-11$ & $10-17$ & $13-23$ & $16-29$ \tabularnewline
\hline
\end{tabular}
\end{center}
\label{tb:mbrr}
\end{table}

\textbf{Remark.} From~\eqref{eq:mscr} and~\eqref{eq5}, we see that the cross-rack repair bandwidth of MSRR codes and MBRR codes decreases along with the increase of $k$, given the same parameters $B,d,m$. If $kr/n$ is an integer, we have $kr/n=\lfloor
(k+i)r/n\rfloor$ for $i=1,2,\ldots,n/r-1$, then the cross-rack repair
bandwidth of MSRR$(n,k+i,r)$ (resp. MBRR$(n,k+i,r)$) codes is strictly less
than that of MSRR$(n,k,r)$ (resp. MBRR$(n,k,r)$) codes. Note that the
cross-rack repair bandwidth of MSRR (resp. MBRR) codes is equal to that of the
minimum storage (resp. bandwidth) codes in \cite{prakash2017} when $kr/n$
is an integer.
Therefore, the construction of MSRR (resp. MBRR) codes when $kr/n$ is not an
integer is necessary and important, as the codes have less cross-rack repair
bandwidth.  We will give the exact-repair constructions for MSRR codes and
MBRR codes in Section~\ref{sec:cons} and Section~\ref{sec:MBRRcons},
respectively.


If we directly employ an RC$(n,k,d')$ in rack-based storage, then we can obtain the trade-off curve between the storage
$\alpha$ and cross-rack repair bandwidth $\gamma'$ of RC$(n,k,d')$ as
the following theorem by Theorem 1 in \cite{dimakis2010}.
\begin{theorem}
If we directly employ an RC$(n,k,d')$ in a rack-based storage system,
i.e., we repair a failed node by downloading $\beta'$ symbols from each of the $d'=dn/r+n/r-1$
helper nodes (including $n/r-1$ nodes in the host rack and $dn/r$ other nodes), then the trade-off
between the smallest storage $\alpha'^{*}(\gamma')$ and cross-rack repair bandwidth $\gamma'$ is
\begin{equation}
\alpha'^{*}(\gamma')=\left\{
\begin{array}{ll}
\frac{B}{k}, &\text{ } \gamma' \in [f'(0),+\infty)\\
\frac{B-g'(i)f'(i)}{k-i},&\text{ } \gamma'\in[\frac{dn}{d'r}f'(i),\frac{dn}{d'r}f'(i-1)),
\end{array}
\right.
\label{eq:rc-min-storage}
\end{equation}
for $i=1,2,\ldots,k-1$, where
\[
g'(i)=i\frac{2d'-2k+i+1}{2d'},
\]
\[
f'(i)=\frac{2Bd'}{2k(d'-k+1)+i(2k-i-1)}.
\]
\label{thm:rc}
\end{theorem}
By Theorem~\ref{thm:rc}, the cross-rack repair bandwidths of RC at MSR point
and MBR point are
\[
\gamma'_{\mathrm{MSR}}=\frac{Bdn/r}{k(dn/r+n/r-k)},
\]
and
\[
\gamma'_{\mathrm{MBR}}=\frac{2Bdn/r}{k(2dn/r+2n/r-k-1)},
\]
respectively. The next theorem shows that MSRR codes (resp. MBRR codes) have less cross-rack repair bandwidth than MSR codes (resp. MBR codes) for most of the parameters.
\begin{theorem}
Let $d'=dn/r+n/r-1$. If $kr/n$ is an integer, then MSR$(n,k,d')$ codes have the same cross-rack
repair bandwidth as MSRR$(n,k,d)$ codes. If $kr/n$ is not an integer, then MSRR$(n,k,d)$ codes have less cross-rack
repair bandwidth than MSR$(n,k,d')$ codes. If $kr/n$ is an integer and $k/n>2/r$,
then MBRR$(n,k,d)$ codes have less cross-rack repair bandwidth than MBR$(n,k,d')$ codes.
\label{thm:cmpr}
\end{theorem}
\begin{proof}
Recall that $m=\lfloor
\frac{kr}{n}\rfloor$. When $kr/n$ is an integer, we have $kr/n=m$ and
\begin{align*}
\gamma'_{\mathrm{MSR}}=\frac{Bdn/r}{k(dn/r+n/r-k)}
=\frac{Bd}{k(d+1-kr/n)},
\end{align*}
which is equal to the cross-rack repair bandwidth in~\eqref{eq:mscr} of
MSRR codes, as $kr/n=m$.
If $kr/n$ is not an integer, we have $kr/n>m$ and we thus obtain that
the cross-rack repair bandwidth of MSRR codes is less than that of MSR codes.
Recall that the cross-rack repair bandwidth of MBRR codes is $\gamma_{\mathrm{MBRR}}$ in~\eqref{eq5}.
If $kr/n$ in an integer, we have $m=kr/n$.
MBRR codes have less cross-rack repair bandwidth than MBR codes, if and only if
\begin{align*}
&\gamma'_{\mathrm{MBR}}-\gamma_{\mathrm{MBRR}}\\
=&\frac{2Bdn/r}{k(2dn/r+2n/r-k-1)}-\frac{Bd}{(k-m)d+m(d-\frac{m-1}{2})}\\
=&\frac{2Bdn}{k(2dn+2n-kr-r)}-\frac{2Bd}{2kd-m(m-1)}\\
=&\frac{2Bd(2nkd-nm(m-1)-k(2dn+2n-kr-r))}{k(2dn+2n-kr-r)(2kd-m(m-1))}\\
=&\frac{2Bd(k^2r+kr-2kn-m^2n+mn)}{k(2dn+2n-kr-r)(2kd-m(m-1))}\\
=&\frac{2Bd(k^2r+kr-2kn-\frac{k^2r^2}{n}+kr)}{k(2dn+2n-kr-r)(2kd-m(m-1))} \text{ by } m=\frac{kr}{n}\\
=&\frac{2Bd(k^2r+2kr-2kn-\frac{k^2r^2}{n})}{k(2dn+2n-kr-r)(2kd-m(m-1))}\\
=&\frac{2Bd(\frac{1}{n}-\frac{2}{kr})}{k(2dn+2n-kr-r)(2kd-m(m-1))}k^2r(n-r)>0.
\end{align*}
Therefore, we can obtain that MBRR codes have less cross-rack repair bandwidth than MBR codes, if and only if $k/n>2/r$. That is to say, if the code rate is not too low,
the cross-rack repair bandwidth of MBRR codes is strictly less than that of MBR codes.
\end{proof}

\begin{figure}
\centering
\vspace{-2em}
\includegraphics[width=0.9\linewidth]{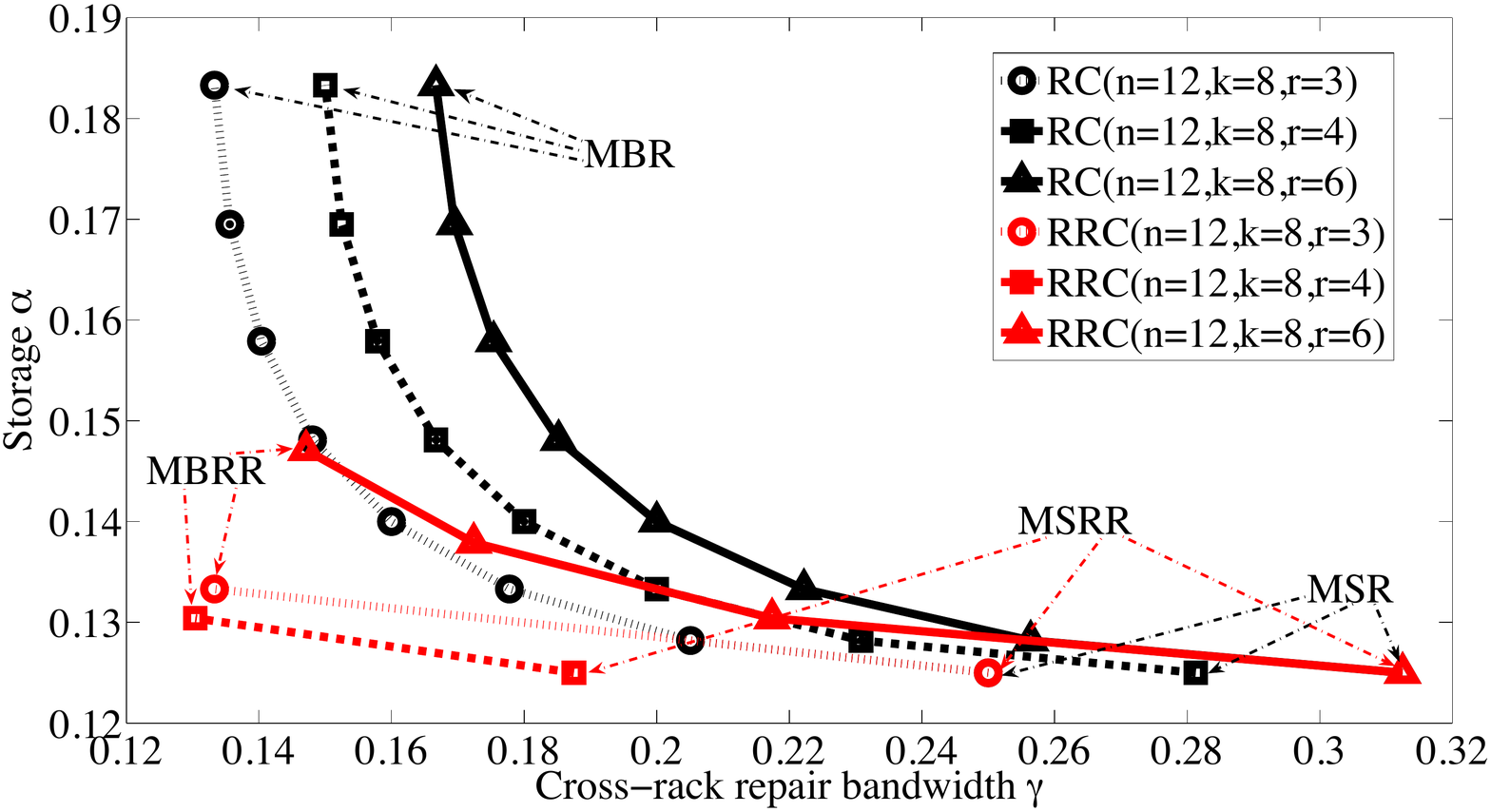}
\caption{Optimal trade-off curve between storage and cross-rack repair bandwidth for RRC and RC when $n=12,k=8$ and $r=3,4,6$ and $d=r-1$. When $(n,k,r)=(12,8,4)$, the cross-rack repair bandwidth of MSRR codes and MSR codes is $\gamma_{\mathrm{MSRR}}=0.1875$ and $\gamma_{\mathrm{MSR}}=0.2813$, respectively; the storage and cross-rack repair bandwidth of MBRR codes and MBR codes are $(\alpha_{\mathrm{MBRR}},\gamma_{\mathrm{MBRR}})=(0.1304,0.1304)$ and $(\alpha_{\mathrm{MBR}},\gamma_{\mathrm{MBR}})=(0.1833,0.1500)$, respectively.}
\label{storage-bandwidth}
\vspace{-0.5em}
\end{figure}

For $B=1$, $n=12$, $k=8$, $r=3,4,6$ and $d=r-1$, the trade-off curves of RRC and RC when
$d'=dn/r+n/r-1$ are shown in Fig.~\ref{storage-bandwidth}.

We have several observations. First, the cross-rack repair bandwidth of RC
increases as $r$ increases under the same storage. Second, unlike RC, the
cross-rack repair bandwidth of RRC when $r=4$ is less than that of $r=3$ under
the same storage. In general, if $kr/n$ is an integer, then the cross-rack
repair bandwidth of RRC increases as $r$ increases, as in RC. However, the
cross-rack repair bandwidth of RRC$(n,k,r')$ is strictly less than that of
RRC$(n,k,r)$ when $kr'/n$ is not an integer and $kr/n=\lfloor kr'/n \rfloor$.
Third, MSRR codes have less cross-rack repair bandwidth than MSR codes
for the same parameters except two points when $r=3,6$. In fact, only
when $kr/n$ is an integer, the cross-rack repair bandwidths of MSRR and MSR
codes are the same according to Theorem~\ref{thm:cmpr}. Finally, MBRR codes
have less cross-rack repair bandwidth than MBR codes for $(n,k,r)=(12,8,4)$
and $(n,k,r)=(12,8,6)$, and have the same cross-rack repair bandwidth as MBR
codes for $(n,k,r)=(12,8,3)$.  By Theorem~\ref{thm:cmpr}, if $kr/n$ is an
integer and $k/n>2/r$, then MBRR codes have less cross-rack repair bandwidth
than MBR codes. Therefore, the results of $(n,k,r)=(12,8,6)$ in
Fig.~\ref{storage-bandwidth} fit well with Theorem~\ref{thm:cmpr}. While
$kr/n$ is not an integer for $(n,k,r)=(12,8,4)$, the cross-rack repair
bandwidth of MBRR codes is less than that of MBR codes. It is interesting to note that
the storage of MBRR codes is strictly less than that of MBR codes for all the
evaluated parameters.
In the rest of the paper, we will focus on the exact-repair constructions of
MSRR codes and MBRR codes.

\section{Exact-Repair Constructions for MSRR Codes}
\label{sec:cons}

This section presents \emph{systematic} constructions of exact repair MSRR
codes. Systematic codes are codes such that the $k\alpha$ uncoded symbols are
stored in $k$ nodes. Suppose that the first $k$ nodes are \emph{data nodes}
that store the uncoded symbols and the last $n-k$ nodes are \emph{coded nodes}
that store the coded symbols.
In Section~\ref{sec:alpha1}, the construction is for
$\alpha=1$ and any $(n,k)$. The construction in Section~\ref{sec:alpha1}
has optimal cross-rack repair bandwidth for any data node and any coded node. In
Section~\ref{msrr-non-integer}, the construction is for
$\alpha n/r\geq m+\alpha t$.  Note that the construction in
Section~\ref{msrr-non-integer} has optimal cross-rack repair bandwidth only
for any data node.
If $kr$ is a multiple of $n$, then $k$ data nodes are replaced in the first
$m$ racks that are called \emph{data racks}, and $n-k$ coded nodes are placed
in the last $r-m$ racks that are called \emph{coded racks}. If $kr$ is not a
multiple of $n$, then the first $m$ racks are data racks, the last $r-m-1$
racks are coded racks, and rack $m+1$ is called \emph{hybrid rack} that
contains $$t:=k \mod (n/r)$$ data nodes and $n/r-t$ coded nodes. MSRR codes
with $kr/n$ being an integer are called \emph{homogeneous MSRR codes}, while
MSRR codes with $kr/n$ being a non-integer are called \emph{hybrid MSRR codes}.

We assume $\beta=1$, and we can extend the construction to $\beta\neq 1$
easily, as in the construction of MSR codes. When $\beta=1$, we have
$$\alpha=d-m+1,B=k(d-m+1).$$
By Theorem~\ref{thm:cmpr}, MSR codes have the same
cross-rack repair bandwidth as homogeneous MSRR codes, i.e., all the existing constructions
of MSR codes can be directly applied to MSRR codes but with much less
intra-rack repair bandwidth, when $kr/n$ is an integer. As the existing construction of MSR codes can
support all the parameters \cite{goparaju2017}, it is not necessary to
exploit the construction of homogeneous MSRR codes. Therefore, we will focus
on the construction of hybrid MSRR codes in the rest of the section.
We first present the construction of MSRR codes for $\alpha=1$.
Then, we give a construction of hybrid MSRR codes for $\alpha n/r
\geq m+\alpha t$ and $\alpha>1$. Our constructions employ \emph{interference
alignment}, which is similar to the concept of \emph{common eigenvector}
that has been used in the construction of exact repair MSR
codes \cite{shah2010,suh2010}.
We first introduce some notation used in this
section before giving the constructions.

\subsection{Notation}

A data file is represented by $B$ \emph{data symbols}
$\mathbf{s}=[s_1,s_2,\ldots,s_B]$ in finite field $\mathbb{F}_q$. Let $\mathbf{s}\mathbf{Q}_{i,h}$ be the
$\alpha$ coded symbols stored in node $i$ and rack $h$ for $i=1,2,\ldots,n/r$
and $h=1,2,\ldots,r$, where $\mathbf{Q}_{i,h}$
is the $B\times \alpha$ \emph{encoding matrix}. The encoding matrix
$\mathbf{G}_h$ of rack $h$ is defined as,
\[\mathbf{G}_h=\begin{bmatrix}\mathbf{Q}_{1,h}& \mathbf{Q}_{2,h} & \cdots & \mathbf{Q}_{n/r,h}\end{bmatrix}.\]
When a node in rack $f$ fails, the new node accesses all $(n/r-1)\alpha$
symbols
in rack $f$, and downloads a coded symbol
from racks $\{h_1,h_2,\ldots,h_d\}\subset\{1,2,\ldots,r\}\setminus\{f\}$ with
\emph{local encoding vector} being $\mathbf{c}_{f,h_i}^T$ for
$i=1,2,\ldots,d$, where $\mathbf{c}_{f,h_i}$ is a row vector
with length $\alpha n/r$. Denote $\mathbf{G}[(i_1,i_2),(j_1,j_2)]$ by the
sub-matrix of $\mathbf{G}$ consisting of rows from
$i_1$ to $i_2$ and columns from $j_1$ to $j_2$, and $\mathbf{G}[(i_1,i_2),(:)]$
and $\mathbf{G}[(:),(j_1,j_2)]$ as two sub-matrices of $\mathbf{G}$ consisting
of rows from $i_1$ to $i_2$ and columns from $j_1$ to $j_2$, respectively.

\subsection{Construction for $\alpha=1$ and Any $n,k$}
\label{sec:alpha1}
We show in the next theorem that any $(n,k)$ MDS code can achieve the minimum cross-rack repair bandwidth.
\begin{theorem}
If $\alpha=1$, then we can repair any one symbol of $(n,k)$ MDS code by downloading
all the other $n/r-1$ symbols in the host rack and one symbol from each of $d$ other arbitrary racks.
\end{theorem}
\begin{proof}
When $\alpha=1$, we have $B=k$ and $d=m$.
For notational convenience, let $c_{i,f}$ denote the symbol
stored in node $i$ and rack $f$, where $i=1,2,\ldots,n/r$ and
$f=1,2,\ldots,r$. We need to show that we can recover the symbol $c_{i,f}$ by
downloading $n/r-1$ symbols
\[
c_{1,f},\ldots,c_{i-1,f},c_{i+1,f},\ldots,c_{n/r,f},
\]
from rack $f$ and $d$ symbols from $d$ arbitrary racks $h_1,h_2,\ldots,h_d$,
where $h_1\neq \cdots \neq h_d \in \{1,2,\ldots,r\}\setminus\{f\}$. We first
consider the case where $kr/n$ is not an integer.

Note that the symbols in rack $h_i$ are
$c_{1,h_i},c_{2,h_i},\ldots,c_{n/r,h_i}$, where $i=1,2,\ldots,d$. Since each
rack has $n/r$ symbols, the total number of symbols in racks
$f,h_1,h_2,\ldots,h_{d-1}$ and the first $t$ symbols in rack $h_d$ is
$dn/r+t=k$. We can view the symbol $c_{n/r,h_d}$ as a linear combination of
the $k$ symbols, including $dn/r$ symbols in racks $f,h_1,h_2,\ldots,h_{d-1}$
and $t$ symbols $c_{1,h_d},c_{2,h_d},\ldots,c_{t,h_d}$ in rack $h_d$, i.e.,
\begin{align*}
c_{n/r,h_d}=& \sum_{j=1}^{n/r}c_{j,f}q_j+\sum_{j=1}^{n/r}c_{j,h_1}q_{j+n/r}+
\sum_{j=1}^{n/r}c_{j,h_2}q_{j+2n/r}+\cdots\\
&+\sum_{j=1}^{n/r}c_{j,h_{d-1}}q_{j+(d-1)n/r}+
\sum_{j=1}^{t}c_{j,h_d}q_{j+dn/r},
\end{align*}
where $q_{i}\neq 0$ for $i=1,2,\ldots,k$. Therefore, we can recover the symbol
$c_{i,f}$ by downloading one symbol
\begin{align*}
& c_{n/r,h_d}-\sum_{j=1}^{t}c_{j,h_d}q_{j+dn/r}\\
=& \sum_{j=1}^{n/r}c_{j,f}q_j+
\sum_{j=1}^{n/r}c_{j,h_1}q_{j+n/r}+\sum_{j=1}^{n/r}c_{j,h_2}q_{j+2n/r}+\cdots\\
&+\sum_{j=1}^{n/r}c_{j,h_{d-1}}q_{j+(d-1)n/r},
\end{align*}
from rack $h_d$, one symbol
\[
\sum_{j=1}^{n/r}c_{j,h_i}q_{j+in/r},
\]
from rack $h_i$ for $i\in \{1,2,\ldots,d-1\}$ and $n/r-1$ symbols
\[
c_{1,f},\ldots,c_{i-1,f},c_{i+1,f},\ldots,c_{n/r,f},
\]
from the rack $f$.

Therefore, any one failure in a rack can be repaired by
downloading one symbol from each of the $d$ arbitrary racks and $n/r-1$
symbols from the host rack, when $kr/n$ is not an integer. The repair process
of the failed symbol $c_{i,f}$ with $kr/n$ being an integer can be viewed as a
special case of the above repair process with $t=0$.  This completes the
proof.
\end{proof}
\begin{figure}
\centering
\includegraphics[width=0.5\textwidth]{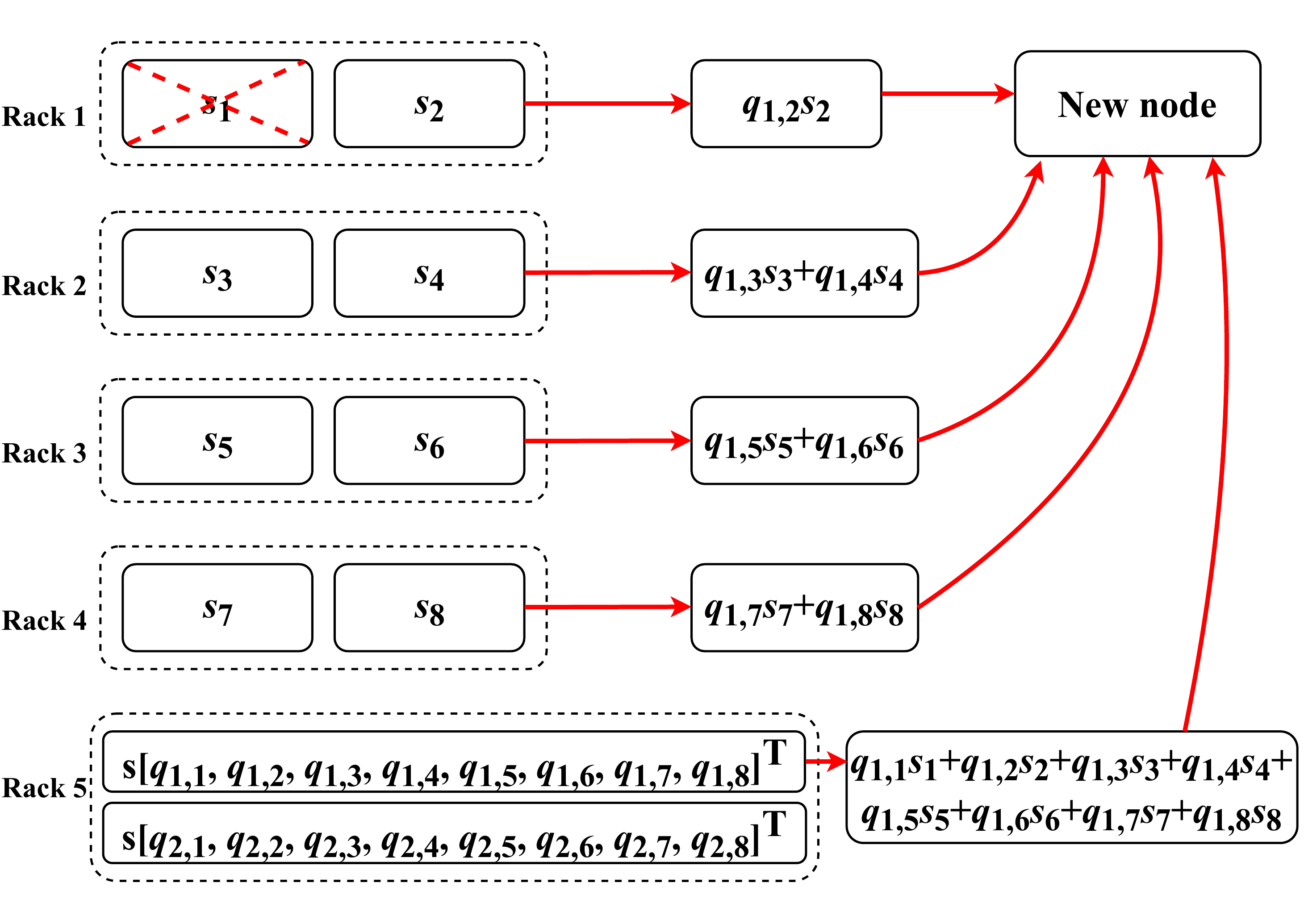}
\caption{Example of homogeneous MSRR code with $(n,k,r,d)=(10,8,5,4)$, the data symbols
$[s_1,s_2,\cdots,s_{8}]$ are denoted by $\mathbf{s}$.
}
\label{fig:example-storage-one}
\vspace{-0.4cm}
\end{figure}

An example in Fig.~\ref{fig:example-storage-one} shows the repair for a data
node. To recover the failure symbol $s_1$, five symbols are downloaded and
only the coded symbol downloaded from rack~5 (is called \emph{desired symbol})
is linearly dependent on $s_1$.
The desired symbol $\sum_{i=1}^{8}q_{1,i}s_i$ is composed of one
\emph{desired component} $q_{1,1}s_1$ which is desirable to
recover the failure symbol and one \emph{interference component}
$\sum_{i=2}^{8}q_{1,i}s_i$. If the interference component is aligned, then
we obtain the desired component $q_{1,1}s_1$ and can repair the failure symbol
if $q_{1,1}\neq 0$. Therefore, the other four coded symbols
downloaded (are called \emph{interference symbols}) are used to align
the interference component. Note that the first construction of DRC in
\cite{hu2017} can be viewed as a special case of $d=r-1$ and $n/(n-k)$ being
an integer.

\subsection{Hybrid MSRR Codes for $\alpha n/r \geq m+\alpha t$}
\label{msrr-non-integer}

\textbf{Idea.} In hybrid MSRR codes, the first $k-t$ data nodes are
placed in the $m$ data racks and the last $t=k \mod (n/r)$ data nodes are
placed in the hybrid rack, each node stores $\alpha=d-m+1$ symbols. In the
repair process of a data node in a data rack, the new node accesses all the
other symbols from the host rack. It downloads (i) $m-1$ coded symbols from
$m-1$ data racks, (ii) one coded symbol from the hybrid rack, and (iii) $d-m$
coded symbols from $d-m$ coded racks. The $d-m+1$ desired symbols are from the
hybrid rack and $d-m$ coded racks.
Note that all the interference symbols that are downloaded from data racks are
independent of the
last $t\alpha$ data symbols. If we want to recover the failed symbols, the
interference component of each of the $d-m+1$ desired symbols should be
independent on the last $t\alpha$ data symbols.
To simultaneously align all the interference components for the $\alpha$
desired symbols, we need to carefully choose the encoding matrices and
introduce the concept of \emph{orthogonal vector}.
We note that in our construction of
hybrid MSRR codes, we can recover the data node by $d$ specific helper racks,
not $d$ arbitrary helper racks, with optimal cross-rack repair bandwidth.


\textbf{Construction.} Before giving the construction, we should introduce
some notation.
The row vectors $\mathbf{v}_1,\mathbf{v}_2,\ldots,\mathbf{v}_{m}$ are
orthogonal of length $m$. The vectors $\mathbf{u}_1,\mathbf{u}_2,\ldots,\mathbf{u}_{\alpha}$
have size $1\times \alpha n/r$. Let $\mathbf{E}_1,\mathbf{E}_2,\ldots,\mathbf{E}_{m}$
be $\alpha n/r\times m$ matrices, $\mathbf{F}_1,\mathbf{F}_2,\ldots,\mathbf{F}_{\alpha}$
be $(B-m\alpha n/r)\times m$ matrices and $\mathbf{R}_1,\mathbf{R}_2,\ldots,\mathbf{R}_{\alpha}$
be $B\times (\alpha n/r-m)$ matrices.

For $h=m+1,m+2,\ldots,r$, the encoding matrix $\mathbf{G}_h$ is given as
\begin{align*}
\mathbf{G}_{h}=\left[
\begin{array}{c:c}
\mathbf{u}_{h-m}^{T}\mathbf{v}_1+\lambda_{h-m,1}\mathbf{E}_1 &\\
\mathbf{u}_{h-m}^{T}\mathbf{v}_2+\lambda_{h-m,2}\mathbf{E}_2 &\\
\vdots & \mathbf{R}_{h-m} \\
\mathbf{u}_{h-m}^{T}\mathbf{v}_{m}+\lambda_{h-m,m}\mathbf{E}_{m} &\\
\mathbf{F}_{h-m} &
\end{array}
\right],
\end{align*}
where the matrix $\mathbf{R}_i$ for $i=1,2,\ldots,\alpha$ is given as follows. The matrix $\mathbf{R}_1$ is
\[
\mathbf{R}_1=\left[
\begin{array}{c:c}
 0_{(B-\alpha\cdot t)\times \alpha \cdot t} &\mathbf{L}_1 \\
I_{\alpha \cdot t\times \alpha \cdot t}&
\end{array}
\right],
\]
where $0_{(B-\alpha\cdot t)\times \alpha \cdot t}$ is a $(B-\alpha\cdot t)\times
\alpha \cdot t$ zero matrix, $I_{\alpha \cdot t\times \alpha \cdot t}$ is an
$\alpha \cdot t\times \alpha \cdot t$ identity matrix and $\mathbf{L}_1$ is a
$B\times (\alpha n/r-m-\alpha \cdot t)$ matrix. Therefore, the parameters
should satisfy $\alpha n/r \geq m+\alpha  t$.
For $i=2,3,\ldots,\alpha$, $\mathbf{R}_i$ is
\[\mathbf{R}_i=\begin{bmatrix}
\mathbf{x}_i^T\mathbf{y}_1+\mathbf{D}_{i,1}\\
\mathbf{x}_i^T\mathbf{y}_2+\mathbf{D}_{i,2}\\
\vdots \\
\mathbf{x}_i^T\mathbf{y}_{m}+\mathbf{D}_{i,m}\\
\mathbf{C}_i
\end{bmatrix},
\]
where $\mathbf{y}_1,\mathbf{y}_2,\ldots,\mathbf{y}_{m}$ are $m$ orthogonal vectors of
length $\alpha n/r-m$, $\mathbf{x}_2,\mathbf{x}_3,\ldots,\mathbf{x}_{\alpha}$ are
$\alpha-1$ vectors of length $\alpha n/r$, $\mathbf{D}_{i,1},\mathbf{D}_{i,2},\ldots,\mathbf{D}_{i,m}$
are matrices of size $\alpha n/r\times (\alpha n/r-m)$ and $\mathbf{C}_i$ is a
matrix of size $(B-m \alpha n/r)\times (\alpha n/r-m)$.

The following requirement should be satisfied. For $f=1,2,\ldots,m$ and
$i=2,3,\ldots,\alpha$, there exist non-zero elements $\lambda_{i,j}^{'}$ for
$j=1,\ldots,f-1,f+1,\ldots,\alpha$ such that the equations
in~\eqref{eq:req1},~\eqref{eq:req2} hold and all the sub-matrices
$\mathbf{M}_1[(:),(1+(\ell-1)\alpha,\ell\alpha)]$ of the matrix
$\mathbf{M}_1$ in~\eqref{eq:req3} are non-singular for $\ell=1,2,\ldots,n/r$.
All the sub-matrices
$\mathbf{M}_2[(1+(i-1)\alpha,i\alpha),(1,\alpha)]$ of the matrix $\mathbf{M}_2$ in~\eqref{eq:req5} are
non-singular for $i=1,2,\ldots,t$. The above condition is called the \emph{repair condition}.

\begin{equation}
\mathbf{D}_{i,j}\mathbf{y}_f^T=\lambda_{i,j}^{'}\mathbf{E}_j\mathbf{v}_f^T
\label{eq:req1}
\end{equation}
\begin{equation}
\mathbf{F}_i\mathbf{v}_f^T=\mathbf{C}_i\mathbf{y}_f^T
\label{eq:req2}
\end{equation}
\begin{equation}
\mathbf{M}_1=\begin{bmatrix}
\mathbf{v}_f(\mathbf{u}_{1}^{T}\mathbf{v}_f+\lambda_{1,f}\mathbf{E}_f)^T\\
\mathbf{v}_f(\mathbf{u}_{2}^{T}\mathbf{v}_f +\lambda_{2,f}\mathbf{E}_f)^T-\mathbf{y}_f(\mathbf{x}_2^T\mathbf{y}_f +\mathbf{D}_{2,f})^T \\
 \vdots \\
\mathbf{v}_f(\mathbf{u}_{\alpha}^{T}\mathbf{v}_f+\lambda_{\alpha,f}\mathbf{E}_f )^T-\mathbf{y}_f(\mathbf{x}_{\alpha}^T\mathbf{y}_f +\mathbf{D}_{\alpha,f})^T
\end{bmatrix}
\label{eq:req3}
\end{equation}
\begin{equation}
\small
\mathbf{M}_2=[
[\mathbf{F}_1 \text{ } \mathbf{L}_1[(m\alpha \frac{n}{r}+1,B),(:)] ]\mathbf{a}^T \text{ } [\mathbf{F}_2 \text{ } \mathbf{C}_2]\mathbf{c}_2^T \text{ }
\cdots \text{ }
[\mathbf{F}_{\alpha} \text{ } \mathbf{C}_{\alpha}]\mathbf{c}_{\alpha}^T
]
\label{eq:req5}
\end{equation}
The proposed codes satisfy $(n,k)$ recovery property if and only if the file
can be retrieved from any $k$ nodes. This is equivalent to that all the
$\alpha\ell\times \alpha\ell$ sub-matrices of the matrix
\begin{equation}\begin{bmatrix}\mathbf{G}_{m+1}[(:),(1,(n/r-t)\alpha)] & \mathbf{G}_{m+2} & \cdots & \mathbf{G}_{r}
\end{bmatrix}
\label{eq:all-encoding-matrices2}
\end{equation}
consisting of rows $i_1,i_1+1,\ldots,i_1+\alpha-1,\ldots,i_\ell,i_\ell+1,\ldots,i_\ell+\alpha-1$
and columns $j_1,j_1+1,\ldots,j_1+\alpha-1,\ldots,j_\ell,j_\ell+1,\ldots,j_\ell+\alpha-1$ are non-singular for
\begin{align*}
&i_1\neq \ldots \neq i_\ell\in\{1,\alpha+1,\ldots,(k-1)\alpha+1\},\\
&j_1 \neq \ldots \neq j_\ell\in\{1,\alpha+1,\ldots,(n-k-1)\alpha+1\},
\end{align*}
where $\ell=1,2,\ldots,\min\{k,n-k\}$. The above requirement
is called the \emph{fault tolerance condition}.

\textbf{Repair.} If a data node in rack $f\in\{1,2,\ldots,m\}$ fails, the new node downloads a desired symbol
\begin{align*}
&\mathbf{s}\mathbf{G}_{m+1}[(1,B),(1,m)]\mathbf{v}_f^T-\mathbf{s}\mathbf{R}_1[(1,B),(1,\alpha t)] \mathbf{F}_1 \mathbf{v}_f^T \\
=&\mathbf{s}\begin{bmatrix}\lambda_{1,1}\mathbf{E}_1 \mathbf{v}_f^T\\
\vdots \\
\lambda_{1,f-1}\mathbf{E}_{f-1}\mathbf{v}_f^T \\
\mathbf{u}_{1}^{T}\mathbf{v}_f \mathbf{v}_f^T+\lambda_{1,f}\mathbf{E}_f \mathbf{v}_f^T \\
\lambda_{1,f+1}\mathbf{E}_{f+1}\mathbf{v}_f^T \\
\vdots \\
\lambda_{1,m}\mathbf{E}_{m} \mathbf{v}_f^T\\
0_{(B-m \alpha n/r) \times 1}
\end{bmatrix},
\end{align*}
from the relayer node in rack $m+1$, one desired symbol
\begin{align*}
&\mathbf{s}\mathbf{G}_{m+i}[(1,B),(1,m)]\mathbf{v}_f^T-\mathbf{s}\mathbf{R}_i\mathbf{y}_f^T \\
=&\mathbf{s}\begin{bmatrix}(\lambda_{i,1}-\lambda_{i,1}^{'})\mathbf{E}_1 \mathbf{v}_f^T\\
\vdots \\
(\lambda_{i,f-1}-\lambda_{i,f-1}^{'})\mathbf{E}_{f-1}\mathbf{v}_f^T \\
\mathbf{u}_{i}^{T}\mathbf{v}_f \mathbf{v}_f^T+\lambda_{i,f}\mathbf{E}_f \mathbf{v}_f^T-(\mathbf{x}_i^T\mathbf{y}_f \mathbf{y}_f^T+\mathbf{D}_{i,f}\mathbf{y}_f^T) \\
(\lambda_{i,f+1}-\lambda_{i,f+1}^{'})\mathbf{E}_{f+1}\mathbf{v}_f^T \\
\vdots \\
(\lambda_{i,m}-\lambda_{i,m}^{'})\mathbf{E}_{m} \mathbf{v}_f^T\\
0_{(B-m\alpha n/r)\times 1}
\end{bmatrix},
\end{align*}
from the relayer node in rack $i+m$ for $i=2,\ldots,\alpha$, and $m-1$
interference symbols from racks $\{1,2,\ldots,m\}\setminus \{f\}$ with local encoding vectors being
\[
\mathbf{E}_1 \mathbf{v}_f^T, \cdots,
\mathbf{E}_{f-1}\mathbf{v}_f^T,
\mathbf{E}_{f+1}\mathbf{v}_f^T,
\cdots,
\mathbf{E}_{m} \mathbf{v}_f^T ,
\]
respectively. Therefore, we obtain $\alpha$ coded symbols which are the
multiplication of the matrix in~\eqref{eq:req3} and
\[\begin{bmatrix}
s_{(f-1)\alpha n/r+1} & s_{(f-1)\alpha n/r+2} & \cdots & s_{f\alpha n/r}
\end{bmatrix}\]
by subtracting the coded symbols downloaded in racks $m+1,\ldots,m+\alpha$
from the symbols downloaded in racks $\{1,\ldots,f-1,f+1,\ldots,m\}$.
The failure $\alpha$ symbols can be recovered, as the corresponding
$\alpha \times \alpha$ sub-matrix of the matrix in~\eqref{eq:req3} is non-singular.


We first review the Schwartz-Zippel Lemma before giving the repair condition
and fault tolerance condition.
\begin{lemma} (Schwartz-Zippel \cite{motwani1995})
Let $Q(x_1,\ldots,x_n)\in \mathbb{F}_q[x_1,\ldots,x_n]$ be a non-zero multivariate polynomial
of total degree $d$. Let $r_1,\ldots,r_n$ be chosen independently and uniformly at
random from a subset $\mathbb{S}$ of $\mathbb{F}_q$. Then
\begin{equation}
Pr[Q(r_1,\ldots,r_n)=0]\leq \frac{d}{|\mathbb{S}|}.
\end{equation}
\end{lemma}

The repair condition and fault tolerance condition can be satisfied if the field size is large enough.

\begin{theorem}
If the field size $q$ is larger than
\begin{equation}
\begin{array}{c}
\alpha\Big(2n/r+t+\sum_{i=1}^{\min\{n-k,k\}}i\dbinom{k}{i}\dbinom{n-k}{i}\Big),
\label{field-size2}
\end{array}
\end{equation}
then there exist encoding matrices $\mathbf{G}_{h}$ for $h=m+1,m+2,\ldots,r$ over
$\mathbb{F}_q$ of hybrid MSRR codes, where the parameters $n,r,m,t,\alpha$ satisfy $\alpha n/r \geq m+\alpha  t$.
\label{thm:hybridcons}
\end{theorem}
\begin{proof}
See Appendix C.
\end{proof}

\begin{table}
\caption{Parameters satisfying the construction of hybrid MSRR codes in Section~\ref{msrr-non-integer}.}
\begin{center}
\begin{tabular}{|c|c|}
\hline
$r$ & $(k,n)$ \tabularnewline
\hline
\hline
3 & (5,9) (7,12) (9-11,12) (8-9,15) (11-14,15) \tabularnewline
\hline
4 & (5,8) (7-8,12) (10-11,12) (9-11,16) (11-14,20) \tabularnewline
\hline
5 & (8,15) (10,15) (11,15) (11,20) (13-15,20) (16-19,25)\tabularnewline
\hline
6 & (7,12) (10-11,18) (13-15,24) (17-19,24) \tabularnewline
\hline
\end{tabular}
\end{center}
\label{constr:prmt2}
\end{table}

From Theorem~\ref{thm:hybridcons}, we obtain that the supported parameters of
the proposed hybrid MSRR codes satisfy $n\geq (m+\alpha t)r/\alpha$. Table~\ref{constr:prmt2}
shows some examples of $d=r-1$ and $r=3,4,5,6$. We can observe from Table~\ref{constr:prmt2}
that we can give the construction of hybrid MSRR codes for most of the parameters.


\textbf{Example.} Take $(n,k,r,d)=(12,8,4,3)$ as an example.
It gives $m=2$, $t=2$, $\alpha=2$ and $B=16$.
The row vectors $\mathbf{v}_1,\mathbf{v}_2$ are orthogonal of length 2, and vectors
$\mathbf{u}_1,\mathbf{u}_2$ are of size $1\times 6$. The matrices $\mathbf{E}_1,\mathbf{E}_2$
have size $6\times 2$, $\mathbf{F}_1,\mathbf{F}_2$ have size $4\times 2$, and
$\mathbf{R}_1,\mathbf{R}_2$ have size $16\times 4$. Then, the encoding matrix $\mathbf{G}_{3}$ is given as
\begin{align*}
\mathbf{G}_{3}=\begin{bmatrix}\mathbf{Q}_{1,3}& \mathbf{Q}_{2,3} & \mathbf{Q}_{3,3}\end{bmatrix}=\left[
\begin{array}{c:c}
0_{12\times 4} & \mathbf{u}_{1}^{T}\mathbf{v}_1+\lambda_{1,1}\mathbf{E}_1 \\
 & \mathbf{u}_{1}^{T}\mathbf{v}_2+\lambda_{1,2}\mathbf{E}_2  \\
 I_{4\times 4} & \mathbf{F}_{1}
\end{array}
\right],
\end{align*}
where $\lambda_{1,1},\lambda_{1,2}$ are two non-zero elements. The encoding matrix $\mathbf{G}_{4}$ is
\begin{align*}
\mathbf{G}_{4}=& \begin{bmatrix}\mathbf{Q}_{1,4}& \mathbf{Q}_{2,4} & \mathbf{Q}_{3,4}\end{bmatrix}\\
=& \left[
\begin{array}{c:c}
\mathbf{u}_{2}^{T}\mathbf{v}_1+\lambda_{2,1}\mathbf{E}_1 & \mathbf{x}_2^T\mathbf{y}_1+\mathbf{D}_{2,1} \\
\mathbf{u}_{2}^{T}\mathbf{v}_2+\lambda_{2,2}\mathbf{E}_2 & \mathbf{x}_2^T\mathbf{y}_2+\mathbf{D}_{2,2} \\
\mathbf{F}_{2} & \mathbf{C}_{2}
\end{array}
\right],
\end{align*}
where $\mathbf{y}_1,\mathbf{y}_2$ are orthogonal vectors of length $4$, $\mathbf{x}_2$ is
of length $6$, $\lambda_{2,1},\lambda_{2,2}$ are two non-zero elements,
$\mathbf{D}_{2,1},\mathbf{D}_{2,2}$ are of size $6\times 4$ and $\mathbf{C}_2$ is a
matrix of size $4\times 4$. Fig.~\ref{fig:example1} shows the example.

\begin{figure}
\centering
\includegraphics[width=0.5\textwidth]{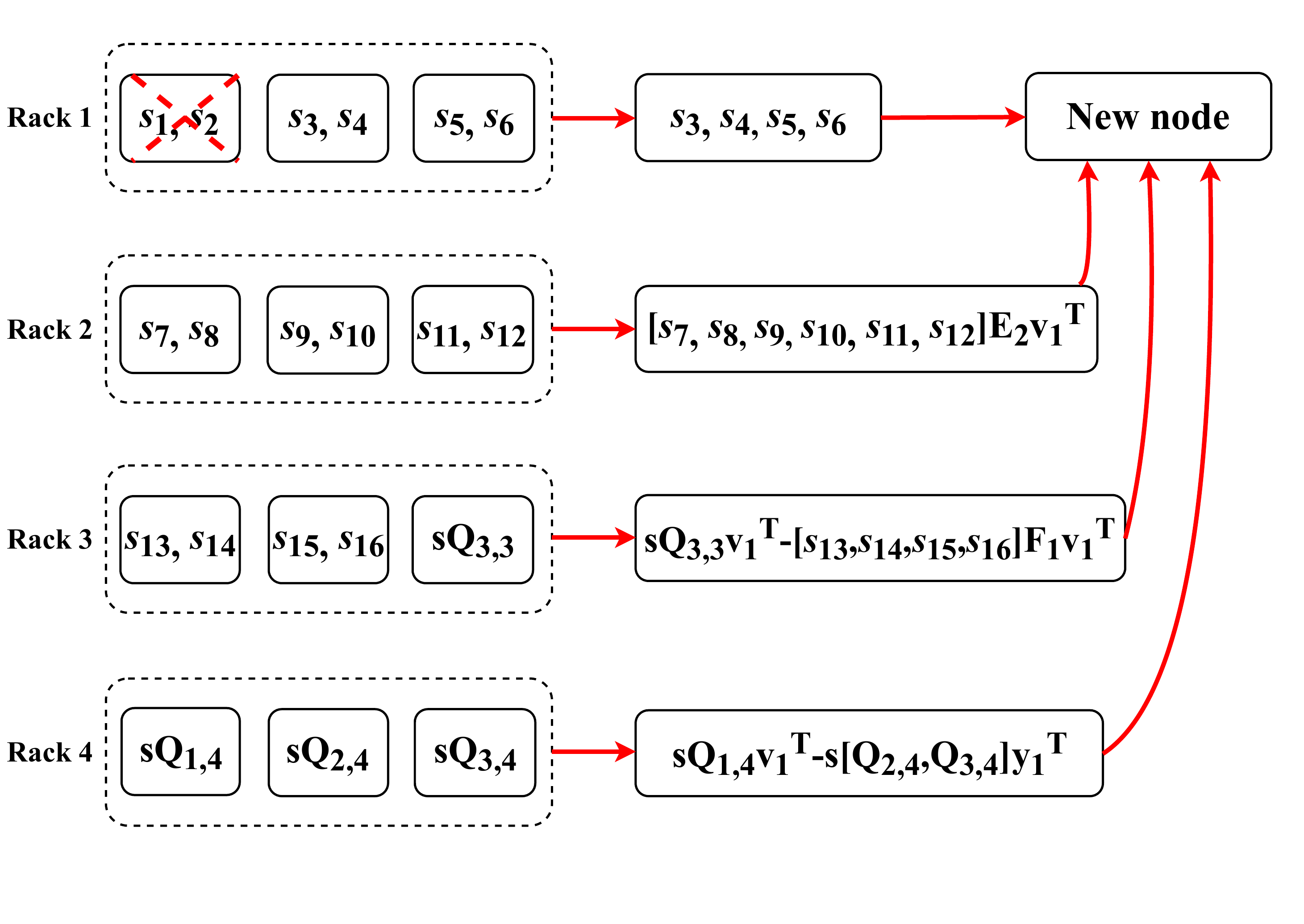}
\caption{Example of hybrid MSRR code with $(n,k,r)=(12,8,4)$. The data symbols are denoted by $\mathbf{s}=[s_1,s_2,\cdots,s_{16}]$.}
\label{fig:example1}
\end{figure}

We can repair two data symbols $s_1,s_2$ in node 1 by downloading
the other four data symbols (the interference symbols) $s_3,s_4,s_5,s_6$ in the
first rack and one symbol (the interference symbol)
\begin{equation}
\lambda_{1,2}\begin{bmatrix}
s_{7} & s_{8} & s_{9} & s_{10} & s_{11} & s_{12}
\end{bmatrix}\mathbf{E}_2\mathbf{v}_{1}^{T},
\label{eq:interf-alignm}
\end{equation}
from rack $2$ and two symbols (the desired symbols)
\begin{align*}
&\mathbf{s}\mathbf{Q}_{3,3}\mathbf{v}_{1}^{T}-\begin{bmatrix}
s_{13} & s_{14} & s_{15} & s_{16}
\end{bmatrix}\mathbf{F}_1\mathbf{v}_{1}^{T}=\\
&\begin{bmatrix}
s_1 & s_2 & \cdots & s_{12}
\end{bmatrix}\begin{bmatrix}
\mathbf{u}_{1}^{T}\mathbf{v}_1\mathbf{v}_1^T+\lambda_{1,1}\mathbf{E}_1 \mathbf{v}_1^T \\
\lambda_{1,2}\mathbf{E}_2 \mathbf{v}_1^T \\
\end{bmatrix},\\
&\mathbf{s}\mathbf{Q}_{1,4}\mathbf{v}_{1}^{T}-\mathbf{s}\begin{bmatrix}
\mathbf{Q}_{2,4} & \mathbf{Q}_{3,4}
\end{bmatrix}\mathbf{y}_{1}^{T}=\\
&\mathbf{s}\begin{bmatrix}
\mathbf{u}_{2}^{T}\mathbf{v}_1\mathbf{v}_1^T+\lambda_{2,1}\mathbf{E}_1 \mathbf{v}_1^T \\
\lambda_{2,2}\mathbf{E}_2 \mathbf{v}_1^T \\
\mathbf{F}_2 \mathbf{v}_1^T \\
\end{bmatrix}-\mathbf{s}\begin{bmatrix}
\mathbf{x}_{2}^{T}\mathbf{y}_1\mathbf{y}_1^T+\mathbf{D}_{2,1} \mathbf{y}_1^T \\
\mathbf{D}_{2,2} \mathbf{y}_1^T \\
\mathbf{C}_2 \mathbf{y}_1^T \\
\end{bmatrix},
\end{align*}
from racks $3,4$. Note that $\mathbf{F}_2 \mathbf{v}_1^T=\mathbf{C}_2 \mathbf{y}_1^T$
according to~\eqref{eq:req2} and $\mathbf{D}_{2,2} \mathbf{y}_1^T=\lambda_{2,2}'\mathbf{E}_2 \mathbf{v}_1^T $
according to~\eqref{eq:req1}. We obtain that the desired symbol downloaded from rack 4 is
\begin{align*}
\mathbf{s}\begin{bmatrix}
\mathbf{u}_{2}^{T}\mathbf{v}_1\mathbf{v}_1^T+\lambda_{2,1}\mathbf{E}_1 \mathbf{v}_1^T-(\mathbf{x}_{2}^{T}\mathbf{y}_1\mathbf{y}_1^T+\mathbf{D}_{2,1} \mathbf{y}_1^T) \\
(\lambda_{2,2}-\lambda_{2,2}')\mathbf{E}_2 \mathbf{v}_1^T \\
\end{bmatrix}.
\end{align*}
By subtracting two desired symbols in racks $3,4$ from the interference symbol
in~\eqref{eq:interf-alignm}, we obtain the following two symbols
\begin{align*}
\small
&\begin{bmatrix}
s_1 & s_2 & \cdots & s_{6}
\end{bmatrix}\cdot \\
&\begin{bmatrix}
\mathbf{u}_{1}^{T}\mathbf{v}_1\mathbf{v}_1^T+\lambda_{1,1}\mathbf{E}_1 \mathbf{v}_1^T \\
\mathbf{u}_{2}^{T}\mathbf{v}_1\mathbf{v}_1^T+\lambda_{2,1}\mathbf{E}_1 \mathbf{v}_1^T-(\mathbf{x}_{2}^{T}\mathbf{y}_1\mathbf{y}_1^T+\mathbf{D}_{2,1} \mathbf{y}_1^T)
\end{bmatrix}^T.
\end{align*}
Therefore, we can recover two data symbols $s_1,s_2$ by first subtracting the above
two symbols from the four interference symbols $s_3,s_4,s_5,s_6$ in the first rack
and then solving the resulting two linear systems, because the $2\times 2$
sub-matrix of the above matrix in the right is non-singular according to~\eqref{eq:req3}.
Nodes 2 and 3 can be recovered similarly.
With the same argument, we can also repair one node in racks 2 and 3.

Although the constructed hybrid MSRR codes only have minimum cross-rack repair
bandwidth for a data node, we can employ the generic
transformation \cite{li2017b} for our hybrid MSRR codes such that each coded node
has the same cross-rack repair bandwidth of
$(n,k=mn/r,d)$ homogeneous MSRR codes.


\section{Exact-Repair Constructions of MBRR Codes for All Parameters}
\label{sec:MBRRcons}

As in the construction of MSRR codes, we also consider the construction of MBRR codes for
$\beta=1$. When $\beta=1$, the parameters of MBRR codes satisfy
\[
B=kd-m(m-1)/2, \alpha=\gamma=d.
\]
Therefore, we want to construct codes with parameters satisfying the above requirement and
the $(n,k)$ recovery property satisfied.

By connecting to any $k$ nodes, we can obtain $kd$ symbols. The $(n,k)$ recovery
property can be satisfied if there exist $B$ independent symbols among the
$kd$ symbols, i.e., there are at most $m(m-1)/2$ dependent symbols in any $k$
nodes. We want to convert the product-matrix (PM) construction of MBR codes
\cite{rashmi2011} into the construction of our MBRR codes. In the following, we present the construction.

The encoding procedure can be described as follows.
\begin{itemize}
\item Divide the $B$ data symbols into two parts, in which the first part has $(k-m)d$
data symbols and the second part has $md-m(m-1)/2$ data symbols.
\item Compute $(n-r-k+m)d$ global coded symbols by encoding all the $B$ data symbols.
Store $(n-r-k+m)d$ global coded symbols and $(k-m)d$ data symbols of the first part
(totally $(n-r)d$ symbols) in the last $(n/r-1)$ nodes of the $r$ racks.
\item Generate $dr$ coded symbols by encoding the first part by a PM-MBR$(r,d,d)$ code.
Divide the generated coded symbols into $r$ groups, each group has $d$ coded symbols.
For each group, take a linear combination for a coded symbol in the group and all
$(n/r-1)d$ symbols stored in the last $(n/r-1)$ nodes in a rack with the encoding
vector being a column vector of length $(n/r-1)d+1$, and the resulting $d$ coded
symbols are stored in the first node of the rack.
\end{itemize}
We show a specific construction as follows.
Denote the row vector
\[
\begin{bmatrix}
s_1 & s_2 & \cdots & s_{(k-m)d}
\end{bmatrix}.
\]
by the first $(k-m)d$ data symbols $s_j$ for $j=1,2,\ldots,(k-m)d$.
Compute $(n-r-k+m)d$ \emph{global coded symbols} by
\begin{align*}
\begin{bmatrix}
c_1 & c_2 & \cdots & c_{(n-r-k+m)d}
\end{bmatrix}=
\begin{bmatrix}
s_1 & s_2 & \cdots & s_{B}
\end{bmatrix}\mathbf{Q},
\end{align*}
where $\mathbf{Q}$ is a $B\times (n-r-k+m)d$ matrix of rank $(n-r-k+m)d$.
Therefore, we obtain $(n-r)d$ symbols
\[
\begin{bmatrix}
s_1 & s_2 & \cdots & s_{(k-m)d} & c_1 & c_2 & \cdots & c_{(n-r-k+m)d}
\end{bmatrix},
\]
which are stored in the last $(n/r-1)$ nodes of the $r$ racks, and are represented by
a $r\times (n/r-1)d$ matrix $\mathbf{M}_1$. Typically, we may choose the matrix
$\mathbf{Q}$ to be Cauchy matrix so that any $k$ out of the $n-r$ nodes (the last
$(n/r-1)$ nodes of the $r$ racks) are sufficient to reconstruct the $B$ data symbols, if $n-r\geq k$.


Create a $d\times d$ data matrix
\[
\mathbf{M}_2:=\begin{bmatrix}
\mathbf{S}_1 & \mathbf{S}_2\\
\mathbf{S}_2^T & \mathbf{0}
\end{bmatrix}.
\]
The matrix $\mathbf{S}_1$ is a symmetric $m\times m$ matrix obtained by first
filling the upper-triangular part by $m(m+1)/2$ data symbols $s_j$, for
$j=(k-m)d+1,(k-m)d+2,\ldots,(k-m)d+m(m+1)/2$, and then obtain the lower-triangular
part by reflection along the diagonal. The rectangular matrix $\mathbf{S}_2$ has
size $m\times (d-m)$, and the entries in $\mathbf{S}_2$ are $m(d-m)$ data symbols
$s_j$, $j=(k-m)d+m(m+1)/2+1,\ldots,B$, listed in some fixed but arbitrary order.
The matrix $\mathbf{S}_2^T$ is the transpose of $\mathbf{S}_2$
and the matrix $\mathbf{0}$ is a $(d-m)\times(d-m)$ all-zero matrix.


Define the matrix $\Phi$ to be a $d\times r$ matrix, with the $i$-th column denoted
by $\phi_i^T$ for $i=1,2,\ldots,r$. Define the matrix $\mathbf{P}$
to be a $(n/r-1)d\times rd$ matrix, with the $\ell$-th column denoted by
$\mathbf{p}_\ell^T$ for $\ell=1,2,\ldots,rd$. For $i=1,2,\ldots,r$,
the $d$ \emph{local coded symbols} stored in the first node in rack $i$ are computed as
\[
(\mathbf{M}_2\phi_i^T)^t+\mathbf{M}_1[(i,i),(1,(n/r-1)d)]\begin{bmatrix}
\mathbf{p}_{(i-1)d+1}^T  & \cdots & \mathbf{p}_{id}^T
\end{bmatrix}.
\]
Note that
\[\begin{bmatrix}
\mathbf{M}_2\phi_1^T & \mathbf{M}_2\phi_2^T & \cdots & \mathbf{M}_2\phi_r^T
\end{bmatrix}\]
can be viewed as the codewords of the PM-MBR$(r,d,d)$ codes
\begin{figure}
\centering
\includegraphics[width=0.5\textwidth]{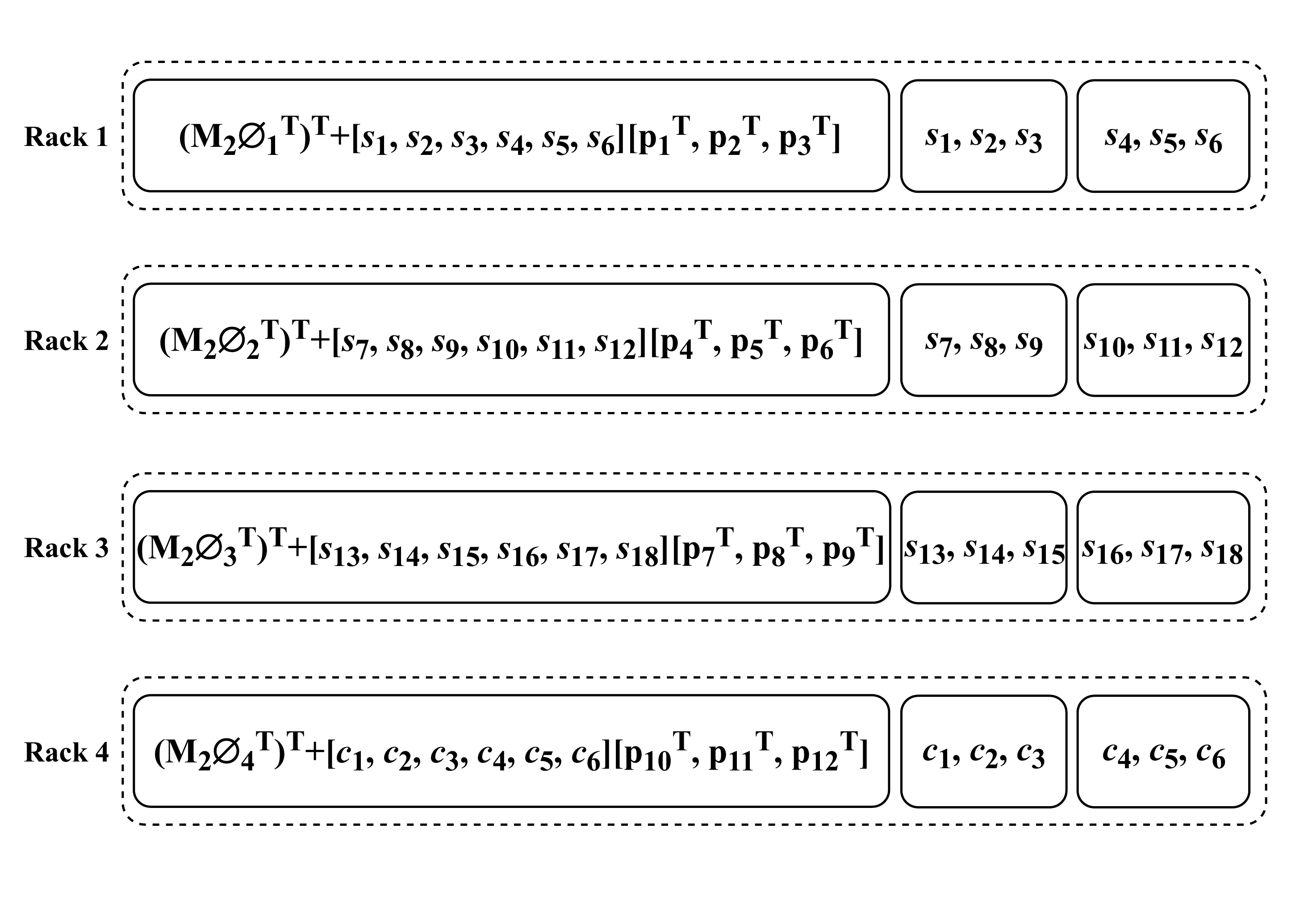}
\caption{Example of MBRR code with $(n,k,r,d)=(12,8,4,3)$.
}
\label{fig:MBRR-example}
\end{figure}

Fig.~\ref{fig:MBRR-example} shows an example of $(n,k,r,d)=(12,8,4,3)$. In the example,
we have $B=23$ data symbols. The first 18 data symbols and 6 global coded symbols are
stored in the last two nodes in each rack, and 4 local coded symbols are stored in
the first node in each rack.

\begin{theorem}
If the field size is larger than
\begin{equation}
B\sum_{i=1}^{\min\{k,r\}}\dbinom{n-r}{k-i}\dbinom{r}{i}.
\label{MBRR-field}
\end{equation}
then any $k$ nodes can recover the $B$ data symbols, and the $\alpha$
symbols stored in any one node can be recovered with optimal cross-rack repair bandwidth of MBRR codes.
\label{thm:mbrr}
\end{theorem}
\begin{proof}
\textbf{File Recovery.} Suppose that a data collector connects to $k$ nodes that are
all from the first $n/r-1$ nodes in the $r$ racks, then we can retrieve the $B$ data
symbols as any square sub-matrix of a Cauchy matrix is non-singular. Consider that a
data collector connects to $k-\ell$ nodes that are from the last $n/r-1$ nodes and
$\ell$ nodes that are from the last node, where $\ell=1,2,\ldots,\min(k,r)$.
The received $kd$ symbols can be represented by the $B\times kd$ encoding matrix.
If we view each entry of $\mathbf{P}$ and $\Phi$ as a non-zero variable, we can
check that there exist a $B\times B$ sub-matrix such that the determinant is a
non-zero polynomial with total degree at most $B$. There are total
\begin{align*}
\sum_{i=1}^{\min\{k,r\}}\dbinom{n-r}{k-i}\dbinom{r}{i}.
\end{align*}
choices. The multiplication of all the determinants is a polynomial with total
degree at most~\eqref{MBRR-field}.
Therefore, we can decode the $B$ data symbols from any $k$ nodes if the field size
is larger than the value in~\eqref{MBRR-field} according to the Schwartz-Zippel Lemma.

\textbf{Repair.} Suppose a node in rack $f$ fails, where $f\in\{1,2,\ldots,r\}$.
The new node connects to any $d$ helper racks $h_i$ for $i=1,2,\ldots,d$.
The relayer node of rack $h_i$ accesses all the symbols stored in the rack,
retrieves $\mathbf{M}_2\phi_{h_i}^T$, computes and sends the coded symbol
\[
\phi_{f}\mathbf{M}_2\phi_{h_i}^T
\]
to the new node. Therefore, the new node obtains $d$ coded symbols
\[
\phi_{f}\mathbf{M}_2\begin{bmatrix}\phi_{h_1}^T & \phi_{h_2}^T & \cdots & \phi_{h_d}^T\end{bmatrix}.
\]
As the left matrix in the above is invertible, the new node can compute the
coded symbols $\phi_{f}\mathbf{M}_2$, as like the repair process of PM-MBR codes.
Then the new node can recover the failure node by accessing all the other symbols in the rack $f$.
\end{proof}
Indeed, the upper bound of field size in Theorem~\ref{thm:mbrr} is exponential in $k$.
However, we may directly check by computer search whether any $k$ nodes can reconstruct
the $B$ data symbols. We have checked by computer search that we can always find
$\mathbf{P}$ and $\Phi$ such that any $k$ nodes can reconstruct the $B$ data symbols
for the example when $(n,k,r,d)=(12,8,4,3)$, when the field size is 11.
We can replace the underlying finite field by a binary cyclic code \cite{hou2016}
for computational complexity reduction.
\section{Comparison}
\label{sec:comp}

In this section, we evaluate cross-rack repair bandwidth for the two extreme
points of RRC, RC and other related codes, such as clustered codes in
\cite{prakash2017} and codes in \cite{sohn2016}. We also discuss the supported
parameters of our exact-repair constructions, the exact-repair constructions in
\cite{prakash2017}, and DRC \cite{hu2016,hu2017}.

\subsection{Cross-rack Repair Bandwidth}
\label{sec:com-repair}
\begin{figure}[t]
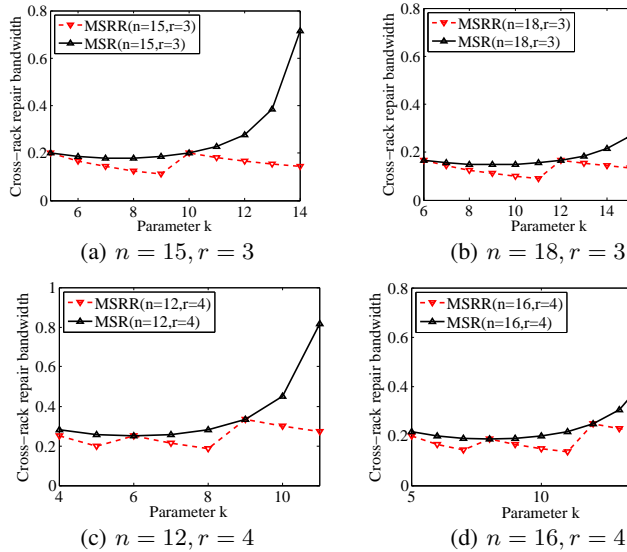

\centering
\begin{tabular}{cc}
\includegraphics[width=0.24\textwidth]{fig//cross-repairMSR315} &
\includegraphics[width=0.23\textwidth]{fig//cross-repairMSR318}
\vspace{-4pt}\\
{\small (a) $n=15,r=3$} &
{\small (b) $n=18,r=3$} \\
\includegraphics[width=0.23\textwidth]{fig//cross-repairMSR412} &
\includegraphics[width=0.25\textwidth]{fig//cross-repairMSR416}
\vspace{-4pt}\\
{\small (c) $n=12,r=4$} &
{\small (d) $n=16,r=4$}
\end{tabular}
\caption{Cross-rack repair bandwidth of MSRR codes and MSR codes when $r=3,4$.}
\label{fig:crmsr}
\end{figure}

\subsubsection{Comparison of MSRR (resp. MBRR) and MSR (resp. MBR)}
According to Theorem~\ref{thm:cmpr}, the cross-rack repair
bandwidth of MSR codes is the same as that of MSRR codes if $kr/n$ is an
integer. Otherwise, if $kr/n$ is not an integer, the cross-rack repair
bandwidth of MSRR codes is strictly less than that of MSR codes.
Fig.~\ref{fig:crmsr} shows the cross-rack repair bandwidth of MSRR codes and
MSR codes when $B=1$, $r=3,4$ and $d=r-1$. The results demonstrate that hybrid
MSRR codes have strictly less cross-rack repair bandwidth than MSR codes and
this advantage increases with $k$.  For example, MSRR codes have 42\% and 83\%
less cross-rack repair bandwidth than MSR codes when $(n,k,r)=(18,11,3)$
and $(n,k,r)=(18,17,3)$, respectively.

By Theorem~\ref{thm:cmpr}, if $k/n>2/r$ and $kr/n$ is an
integer, then MBRR codes have less cross-rack repair bandwidth than MBR codes.
Therefore, if the code rate is not too low, the cross-rack repair bandwidth of
MBRR codes is strictly less than that of MBR codes.

\begin{figure}[t]
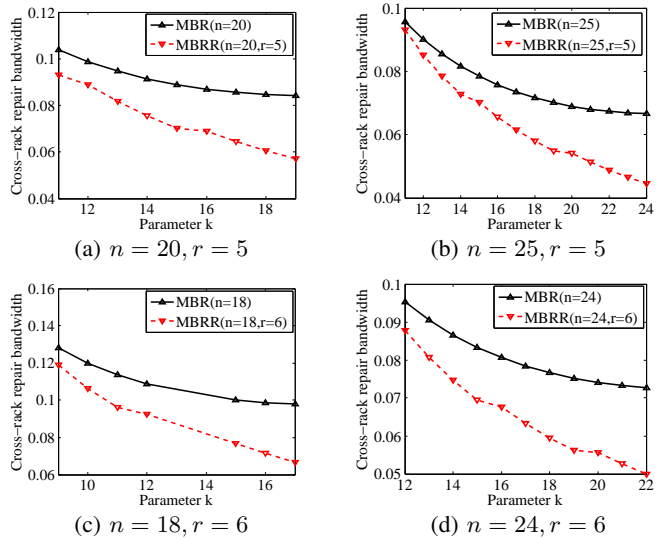

\centering
\begin{tabular}{cc}
\includegraphics[width=0.23\textwidth]{fig//MBRRr5n20} &
\includegraphics[width=0.24\textwidth]{fig//MBRRr5n25}
\vspace{-4pt}\\
{\small (a) $n=20,r=5$} &
{\small (b) $n=25,r=5$} \\
\includegraphics[width=0.23\textwidth]{fig//MBRRr6n18} &
\includegraphics[width=0.24\textwidth]{fig//MBRRr6n24}
\vspace{-4pt}\\
{\small (c) $n=18,r=6$} &
{\small (d) $n=24,r=6$}
\end{tabular}
\caption{Cross-rack repair bandwidth of MBRR codes and MBR codes when $r=5,6$.}
\label{fig:crmbr}
\end{figure}

\begin{figure}[t]
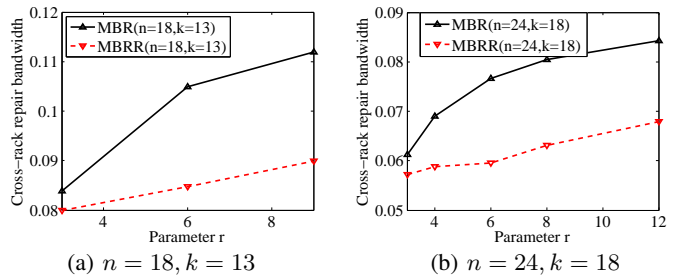

\centering
\begin{tabular}{cc}
\includegraphics[width=0.23\textwidth]{fig//MBRn18k13} &
\includegraphics[width=0.25\textwidth]{fig//MBRn24k18}
\vspace{-4pt}\\
{\small (a) $n=18,k=13$} &
{\small (b) $n=24,k=18$}
\end{tabular}
\caption{Cross-rack repair bandwidth of MBRR codes and MBR codes when $n=18,24$.}
\label{fig:crmbr1}
\end{figure}

Fig.~\ref{fig:crmbr} shows the cross-rack repair bandwidth of
two codes when $B=1$, $r=5,6$ and $d=r-1$. The results demonstrate that MBRR
codes have less cross-rack repair bandwidth for all the evaluated parameters.
Given $r$ and $n$, we note that the differences between MBRR codes and MBR
codes become larger when $k$ increases.  When $n=20$ and $r=5$, the reduction
in the cross-rack repair bandwidth of MBRR codes over MBR codes is from 11\%
to 32\%. When $n=18$ and $r=6$, the reduction is from 7\% to 32\%.

Let $B=1$ and $d=r-1$. For a specific case where
$(n,k)=(18,13)$ (resp. $(n,k)=(24,18)$), Fig.~\ref{fig:crmbr1} shows the
cross-rack repair bandwidth of MBR codes and MBRR codes when $r=3,6,9$ (resp.
$r=3,4,6,8,12$). We have two observations from Fig.~\ref{fig:crmbr1}.  First,
the cross-rack repair bandwidth of MBRR codes is always less than that of MBR
codes.  Second, for given $n$ and $k$, the advantage of the lower cross-rack
repair bandwidth of MBRR codes varies significantly for different values of $r$.  For
example, MBRR codes have 6.8\% reduction of cross-rack repair bandwidth
compared to MBR codes when $(n,k,r)=(24,18,3)$, while the reduction increases
to 21.6\% when $(n,k,r)=(24,18,8)$.

\subsubsection{Comparison of MSRR (resp. MBRR) and Minimum Storage (resp. Bandwidth)
Point of Codes in \cite{prakash2017,sohn2016}}
Two nearest related works are \cite{prakash2017,sohn2016}.
In our model, any $k$ nodes are sufficient to reconstruct the data file, while
in \cite{prakash2017}, any $kr/n$ racks (where $kr/n$ is an integer) can
reconstruct the data file and there may exist $k$ nodes that cannot
reconstruct the data file.
In \cite{prakash2017}, a failed node is recovered by downloading $\alpha$
symbols from each of $\ell$ other nodes in the host rack, and $\beta$ symbols
from each of $d$ other racks. The $\beta$ symbols downloaded from the remote
racks are linear combinations of all the $\alpha n/r$ symbols stored in the
rack. Under the setting of functional repair, it is shown in Theorem 4.1 in
\cite{prakash2017} that the file size is upper bounded by (note that we should
replace $k$ by $kr/n$ and $m$ by $n/r$ in (11) in \cite{prakash2017}):
\[
B\leq \ell \frac{kr}{n}\alpha +(n/r-\ell)\sum_{i=0}^{kr/n-1}\min\{\alpha,\max\{d-i,0\}\beta\}.
\]
If we download all $\alpha(n/r-1)$ symbols in the host rack to repair the failed node, i.e., $\ell=n/r-1$, then the above upper bound is
\begin{align*}
B&\leq (n/r-1) \frac{kr}{n}\alpha +\sum_{i=0}^{kr/n-1}\min\{\alpha,\max\{d-i,0\}\beta\}\\
&=k\alpha +\sum_{i=0}^{kr/n-1}\min\{0,\max\{d-i,0\}\beta-\alpha\}.
\end{align*}
We assume that $d\geq kr/n$. Then
$\max\{d-i,0\}\beta-\alpha=(d-i)\beta-\alpha$, and the above bound is the same
as the bound in~\eqref{thm1} in our Theorem~1 when $kr/n$ is an integer.
Therefore, the cross-rack repair bandwidth of the codes in \cite{prakash2017}
is equal to that of our RRC. According to the remark before
Theorem~\ref{thm:rc}, MSRR$(n,k+i,r)$ codes (resp. MBRR$(n,k+i,r)$ codes)
have strictly less cross-rack repair bandwidth than MSRR$(n,k,r)$
codes (resp. MBRR$(n,k,r)$ codes) for $i=1,2,\ldots,n/r-1$, where $kr/n$ is an
integer. We thus obtain that MSRR$(n,k+i,r)$ codes (resp. MBRR$(n,k+i,r)$
codes) have strictly less cross-rack repair bandwidth than the minimum storage
(resp. bandwidth) $(n,k,r)$ codes in \cite{prakash2017} for
$i=1,2,\ldots,n/r-1$, where $kr/n$ is an integer.
Note that any $k$ nodes can reconstruct the data file in our model, while not
in \cite{prakash2017}. However, the bound of our model is the same as that
in \cite{prakash2017} under the same setting. We have the following observation
from our results and the results in \cite{prakash2017}. All $\alpha(n/r-1)$
symbols in the host rack are necessary to obtain the minimum cross-rack repair
bandwidth.  If we reduce the intra-rack repair bandwidth, i.e., reduce $\ell$,
then the cross-rack repair bandwidth will be increased.

Let $\beta_I$ and $\beta_c$ be the numbers of symbols
downloaded from a helper node in the host rack and the other racks,
respectively. Let $\epsilon=\beta_c /\beta_I$. It is shown in Theorem 3 in
\cite{sohn2016} that the minimum storage overhead, i.e., $\alpha=B/k$, is
achieved if and only if $\epsilon \geq 1/(n-k)$. Therefore, $\epsilon =
1/(n-k)$ is the scenario with minimum cross-rack repair bandwidth when the
minimum storage overhead is imposed.  When $\epsilon=1/(n-k)$, the minimum
storage point of the codes in \cite{sohn2016} is
\[
(\alpha_{\mathrm{MSR}},\gamma_{\mathrm{MSR}})=(\frac{B}{k},\frac{B}{k}\frac{n-n/r}{n-k}),
\]
where $\gamma_{\mathrm{MSR}}$ is the cross-rack repair bandwidth and all $n-1$ surviving nodes are contacted.
We have three observations. First, all $\alpha$ symbols in the node of host rack should be
downloaded to minimize the cross-rack repair bandwidth in the repair. Second, the cross-rack
repair bandwidth of the minimum storage point of the codes in \cite{sohn2016} is the same as
that of original MSR codes for all parameters under the same setting when $\epsilon = 1/(n-k)$.
Third, the cross-rack repair bandwidth of our MSRR codes is strictly
less than that of the minimum storage point of the codes in \cite{sohn2016} when $t\neq 0$. On the other
hand, if $t=0$, i.e., $kr/n$ is an integer, the cross-rack repair bandwidth of MSRR codes
 is equal to that of the minimum storage point of the codes in \cite{sohn2016}.

Consider the cross-rack repair bandwidth of the minimum bandwidth point of
the codes in \cite{sohn2016}. For $\epsilon>0$, set $\beta_c=1$, then $\beta_I=1/\epsilon$.
The minimum bandwidth point of the codes in \cite{sohn2016} is
\[
(\alpha_{\mathrm{MBR}},\gamma_{\mathrm{MBR}})=((n/r-1)/\epsilon +(n-n/r),(n-n/r)),
\]
where $\gamma_{\mathrm{MBR}}$ is the cross-rack repair bandwidth, and the file size is
\[
B=k\alpha_{\mathrm{MBR}}-\frac{1}{2}(\frac{1}{\epsilon}-1)(m(n/r)^2+t^2-k)-\frac{k(k-1)}{2},
\]
according to Proposition 3 of \cite{sohn2018}. We observe that the storage increases as
$\epsilon$ decreases. If we decrease $\epsilon$, then the normalized cross-rack
repair bandwidth will be decreased at the expense of increasing the storage. Note that the
storage $\alpha_{\mathrm{MBR}}$ is larger than the cross-rack repair bandwidth $\gamma_{\mathrm{MBR}}$
in the minimum bandwidth point of the codes in \cite{sohn2016},
while the storage is equal to the cross-rack repair bandwidth in our MBRR codes.

\begin{figure}[t]
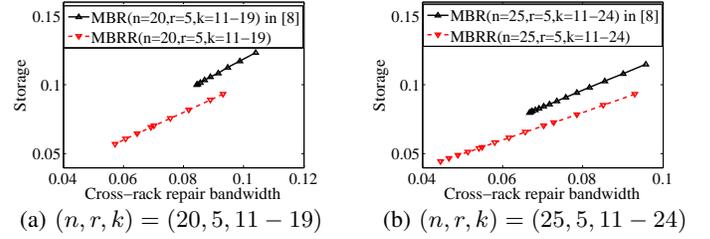

\centering
\begin{tabular}{cc}
\includegraphics[width=0.24\textwidth]{fig//MBRn20r5} &
\includegraphics[width=0.24\textwidth]{fig//MBRn25r5}
\vspace{-4pt}\\
{\small (a) $(n,r,k)=(20,5,11-19)$} &
{\small (b) $(n,r,k)=(25,5,11-24)$}
\end{tabular}
\caption{The trade-off between storage and cross-rack repair bandwidth of MBRR codes
and the minimum bandwidth point of the codes in \cite{sohn2016} when $n=20,25$.}
\label{fig:s-crmbr1}
\end{figure}



Let $\epsilon=1$. Fig.~\ref{fig:s-crmbr1} shows the trade-off between
storage and cross-rack repair bandwidth when $B=1$, $n=20,25$, $r=5$ and $d=4$. We can observe from
Fig.~\ref{fig:s-crmbr1} that both the storage and cross-rack repair bandwidth of
MBRR codes is less than that of the minimum bandwidth point of the codes
in \cite{sohn2016} for all the evaluated parameters.

In conclusion, the cross-rack repair bandwidth of our RRC is strictly less
than that of the codes in \cite{sohn2016} for most of the parameters, and is
the same as that of the codes in \cite{prakash2017} if $kr/n$ is an integer.
Also, RRC can tolerate more failure patterns than the codes in
\cite{prakash2017}. When $kr/n$ is an integer, hybrid MSRR$(n,k+i,r)$ codes
for $i=1,2,\ldots,n/r-1$ have less cross-rack repair bandwidth than MSR$(n,k,r)$
codes and the minimum storage $(n,k,r)$ codes in \cite{prakash2017}.

\subsection{Parameters of Exact Repair MSRR Codes and MBRR Codes}

\label{sec:parameter}
\begin{figure}[t]
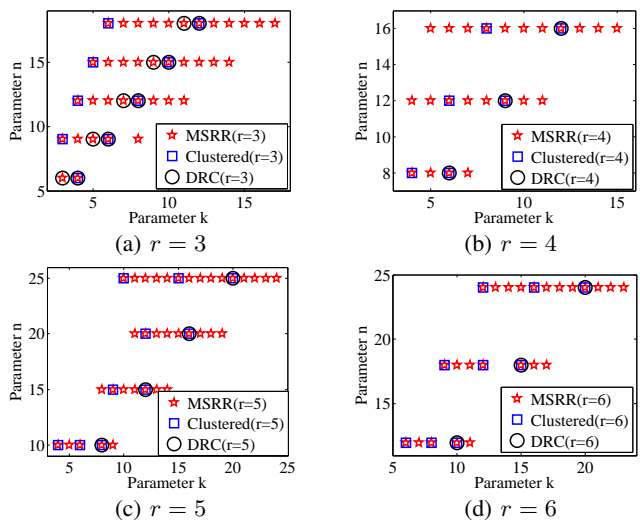

\centering
\begin{tabular}{cc}
\includegraphics[width=0.23\textwidth]{fig//parameterr3} &
\includegraphics[width=0.23\textwidth]{fig//parameterr4}
\vspace{-4pt}\\
{\small (a) $r=3$} &
{\small (b) $r=4$} \\
\includegraphics[width=0.23\textwidth]{fig//parameterr5} &
\includegraphics[width=0.23\textwidth]{fig//parameterr6}
\vspace{-4pt}\\
{\small (c) $r=5$} &
{\small (d) $r=6$}
\end{tabular}
\caption{Supported parameters of MSRR codes, clustered codes in \cite{prakash2017} and DRC when $r=3,4,5,6$.}
\label{fig:r3}
\end{figure}

We now present numerical evaluation of supported parameters of exact-repair
construction for two extreme points of RRC, the codes in
\cite{prakash2017,sohn2016} and DRC in \cite{hu2017}. The first construction
of DRC in \cite{hu2017} can be viewed as a special case of our construction of
MSRR codes in Section~\ref{sec:alpha1} with $n/(n-k)$ being an integer and
$d=r-1$, and the second construction of DRC in \cite{hu2017} only focuses on
the case of $r=3$. In the construction of the codes in \cite{prakash2017},
$kr/n$ should be an integer. The construction of the minimum
storage codes in \cite{sohn2016} is given in \cite{sohn2018a} and it can only
support $r=2$ and $n=2k$. When $r=3,4,5,6$ and $n$ takes different values,
the supported values of $k$ for MSRR codes, clustered codes in
\cite{prakash2017} and DRC are shown in Fig.~\ref{fig:r3}. The results show
that the supported parameters of MSRR codes are much more than those of the
two codes.

Note that the exact-repair construction of clustered codes in
\cite{prakash2017} is based on the existing constructions of MSR codes.  When
$k/n<0.5$, there is a limitation that the storage $\alpha$ is exponential to
$k$ for the existing constructions of MSR codes and the minimum storage
construction of clustered codes in \cite{prakash2017}. However, our
construction of MSRR codes does not have this limitation.

Our construction of MBRR codes can support all the
parameters.  The construction of the minimum bandwidth codes in
\cite{sohn2016} is given in \cite{sohn2018} and it can also support all the
parameters. On the other hand, $kr/n$ should be an integer in the construction
of the minimum bandwidth codes in \cite{prakash2017}. Therefore, our
construction can support more parameters than that in \cite{prakash2017}.
\section{Conclusions and Future Work}
\label{sec:conclu}

In this paper, we study the optimal trade-off between storage and cross-rack
repair bandwidth of rack-based data centers. We propose Rack-aware
Regenerating Codes (RRC) that can achieve the optimal trade-off. We derive
two extreme optimal points, namely the MSRR and MBRR points, and give exact-repair constructions of MSRR codes and MBRR codes.  We show that the
cross-rack repair bandwidth of MSRR codes (resp. MBRR codes) is strictly
less than that of MSR codes (resp. MBR codes) for most of the parameters.
In our system model, all the symbols in the host rack are
downloaded to repair a failed node. One future work is to generalize the
results for more flexible selection of helper nodes in the host rack.  Another
future work is the implementation of RRC in practical rack-based data
centers.

\appendices
\section{Proof of Theorem \ref{thm1}}
\label{pr:thm1}
\begin{proof}
First, we show the following lemma.
\begin{lemma}
If a relayer in a rack is connected to the data collector $\sT$ and not all
the other $n/r-1$ nodes in the rack are connected to $\sT$, then the capacity
of $(\sS,\sT)$-{\em cut} is not the smallest.
\end{lemma}
\begin{proof}
Consider that a relayer $X_{1,1}$ is connected to $\sT$. Since the incoming
edges of $\sT$ all have infinite capacity, we only need to examine the
incoming edges of $\sOut_{1,1}$ and $\sIn_{1,1}$. As $X_{1,1}$ is not a failed
node, the incoming edges of $\sOut_{1,1}$ and $\sIn_{1,1}$ have capacity
$\alpha n/r$ and infinite, respectively. So a relayer without failure
contributes $\alpha n/r$ to the cut. On the other hand, if a relayer
$X_{1,1}^{'}$, which is a failed node, is connected to $\sT$, the incoming edges
of $\sOut_{1,1}'$ and $\sIn_{1,1}'$ have capacity $\alpha n/r$ and
$\alpha (n/r-1)+d\beta$, respectively. The node
$X_{1,1}^{'}$ can contribute $\min\{(n/r-1)\alpha+d\beta,\alpha n/r\}$
symbols to the cut. Recall that each of all the other $n/r-1$ nodes in
rack 1 has an edge which connects to the input node and $\sOut_{1,1}$ with
capacity $\alpha$. All the other $n/r-1$ nodes in rack 1 have no contribution
to the cut whether they are connected to $\sT$ or not. Therefore, if a
relayer is connected to $\sT$, we should connect to all the other $n/r-1$ nodes
in the same rack to $\sT$ to minimize the capacity of the cut.
\end{proof}

Next, we show that there exists an information flow graph $G(n,k,r,d,\alpha,\beta)$
such that $\text{mincut}(G)$ is equal to the right value in~\eqref{eq1}. In the graph,
the relayer nodes $X_{1,1},X_{2,1},\ldots,X_{m,1}$ fail in this order. Each new node
$X_{\ell,1}^{'}$ draws $\alpha$ symbols from each of nodes $X_{\ell,2},X_{\ell,3},\ldots,X_{\ell,n/r}$
and $\beta$ symbols from each of the first $d$ relayer nodes, for $\ell=1,2,\ldots,m$.
Consider the data collector $\sT$ that connects to all nodes in the first $m$ racks and
$k-m n/r$ nodes (except the relayer node) in rack $m+1$. Fig.~\ref{model-example} shows
the graph $G(n,k,r,d,\alpha,\beta)$ when $(n,k,r,d)=(9,5,3,2)$. For each
$\ell \in \{1,2,\ldots, m\}$, rack $\ell$ can contribute $\min\{(n/r-1)\alpha+(d-\ell+1)\beta,n/r\cdot \alpha\}$
to the cut. Therefore, $\text{mincut}(G)$ is the right side in~\eqref{eq1}.

In the following, we show that~\eqref{eq1} must be satisfied for any information flow
graph $G(n,k,r,d,\alpha, \beta)$. Consider that $\sT$ connects to $k$ ``out-vertices'',
which are represented by $\{\sOut_{h,i}:(h,i)\in \mathbb{I}\}$, the cardinality of
$\mathbb{I}$ is $k$. We want to show that the smallest $\text{mincut}(G)$ is at least the right value in~\eqref{eq1}.

Without loss of generality, $\sOut_{h_1,i_1},\ldots,\sOut_{h_{n/r},i_{n/r}}$ are
assumed to be the first $n/r$ out-vertex in the cut. If there is only one vertex
$\sOut_{h_{\ell},i_{\ell}}$ that is a relayer for $i_{\ell}\in\{1,2,\ldots,n/r\}$,
then it can contribute $\min\{(n/r-1)\alpha+d\beta,n/r\cdot \alpha\}$ to the cut and
we select $n/r-1$ vertices to be located in the same rack that have no contribution
to the cut. If the number of relayer is larger than 1, then the contribution is larger
than $\min\{(n/r-1)\alpha+d\beta,n/r\cdot \alpha\}$. If all the vertices
$\sOut_{h_1,i_1},\ldots,\sOut_{h_{n/r},i_{n/r}}$ are not relayers, then they can
contribute $\alpha n/r$ to the cut. Therefore,  the $n/r$ vertices contribute at
least $\min\{(n/r-1)\alpha+d\beta,n/r\cdot \alpha\}$ to the cut.

Now, we assume $\sOut_{h_{n/r+1},i_{n/r+1}},\ldots,\sOut_{h_{2n/r},i_{2n/r}}$ are
the second $n/r$ out-vertices. Similar the above discussion, we have that those
$n/r$ nodes contribute at least $\min\{(n/r-1)\alpha+(d-1)\beta,n/r\cdot \alpha\}$
to the cut. By the same arguments for the $\ell$-th $n/r$ vertices for
$\ell=3,4,\ldots,m$ and the last $k-m n/r$ vertices, we will have that a min-cut
for any information flow graph $G(n,k,r,d,\alpha,\beta)$ is exactly the right value in~\eqref{eq1}.
\end{proof}

\section{Proof of Theorem~\ref{thm:opt-tradeoff}}
\label{pr:opt-tradeoff}
\begin{proof}
We need to solve for $\alpha^{*}(\beta)$ as follows,
\begin{align*}
&\alpha^{*}(n,k,r,\beta)\triangleq \min \alpha \\
&\text{subject to: } k\alpha+\sum_{\ell=1}^{m}\min\{(d-\ell+1)\beta-\alpha,0\}\geq B.
\end{align*}
If $\alpha \leq (d-m+1)\beta$, then we have $k\alpha\geq B$ and $\alpha^{*}(\beta)=B/k$.
If $\alpha \geq d\beta$, we have
\[
k\alpha+(d\beta-\alpha)+((d-1)\beta-\alpha)+\cdots+((d-m+1)\beta-\alpha)\geq B ,
\]
and
\[
\alpha^{*}(\beta)=\frac{Bd}{(k-m)d+m(d-\frac{m-1}{2})}.
\]
For $i=1,2,\ldots,m-1$, if $(d-m+i+1)\beta<\alpha\leq (d-m+i+2)\beta$, then
the capacity is
\begin{align*}
&k\alpha+((d-m+1)\beta-\alpha)+((d-m+2)\beta-\alpha)+\cdots+\\
&((d-m+i+1)\beta-\alpha) \\
=&(k-i-1)\alpha+(i+1)(d-m+i/2+1)\beta.
\end{align*}
The relation of the smallest capacity $\Phi$ and $\alpha$ is as follows
\begin{equation}
\Phi= \left\{ \begin{array}{rcl}
k\alpha,    & & \alpha\in [0, b_0],  \\
(k-1)\alpha+b_0,   &  & \alpha\in (b_0,b_1], \\
\vdots   & & \vdots \\
 \sum_{j=0}^{m-2}b_j+(k-m+1)\alpha,   & & \alpha\in (b_{m-2},b_{m-1}],\\
 \sum_{j=0}^{m-1}b_j+(k-m)\alpha,    & & \alpha\in (b_{m-1},\infty],
\end{array}\right.
\end{equation}
where
\begin{equation}
b_i=\beta(d-m+i+1),
\label{eq:bi}
\end{equation}
for $i=0,1,\ldots,m-1$.
Recall that $\Phi\geq B$, and we can solve for $\alpha^{*}(\beta)$, which is
\begin{equation}
\left\{
\begin{array}{ll}
\frac{B}{k}, &\text{ }  B \in [0,kb_0 ]\\
\frac{B-b_0}{k-1}, &\text{ } B\in(kb_0,b_0+(k-1)b_1]\\
\vdots  \\
\frac{B-\sum_{j=0}^{m-2}b_j}{k-m+1}, &\text{ } B\in(\sum_{j=0}^{m-2}b_j+(k-m+1)b_{m-2}\\
&,\sum_{j=0}^{m-1}b_j+(k-m)b_{m-1}].
\end{array}
\right.
\label{eq4}
\end{equation}
Therefore, if
\[
B\in (\sum_{j=0}^{i-1}b_j+(k-i)b_{i-1},\sum_{j=0}^{i}b_j+(k-i-1)b_{i}],
\]
then
\[
\alpha^{*}(\beta)=\frac{B-\sum_{j=0}^{i-1}b_j}{k-i},
\]
for $i=1,2,\ldots,m-1$.
Recall that $b_i$ is defined in~\eqref{eq:bi}, we compute that
\[
\sum_{j=0}^{i-1}b_j=\sum_{j=0}^{i-1}\beta(d-m+j+1)=\beta i (2d-2m+i+1)/2,
\]
\begin{align*}
&\sum_{j=0}^{i}b_j+(k-i-1)b_i \\
=&\beta (i+1) (d-m+i/2+1)+(k-i-1)\beta(d-m+i+1)\\
=& \beta(2k(d+1)-2km+2ki-i^2-i)/2.
\end{align*}
Then we have the optimal trade-off in the theorem.
\end{proof}
\section{Proof of Theorem~\ref{thm:hybridcons}}
\label{pr:hybirdcons}
\begin{proof}
We view the values of $\mathbf{v}_1,\ldots,\mathbf{v}_m,\mathbf{y}_1,\ldots,\mathbf{y}_m$
and $\lambda_{i,j}$ as constants and the other entries of vectors and matrices as variables.
There are
\begin{align*}
\underbrace{\alpha n/r(\alpha n/r-m)(\alpha-1)^2}_{\mathbf{D}_{i,j}}+\underbrace{(\alpha-1)^2}_{\lambda_{i,j}^{'}}+\underbrace{\alpha mn/r(\alpha-1)}_{\mathbf{E}_j}
\end{align*}
variables, $\alpha mn/r(\alpha-1)^2$ equations and
\begin{align*}
\underbrace{(B-\alpha mn/r)m(\alpha-1)}_{\mathbf{F}_{i,j}}+\underbrace{(B-\alpha m n/r)(\alpha n/r-m)(\alpha-1)}_{\mathbf{C}_{i}}
\end{align*}
variables, $(B-\alpha mn/r)m(\alpha-1)$ equations
in~\eqref{eq:req1} and~\eqref{eq:req2}, respectively.
Note that $\alpha n/r \geq 2m$, we can view all the entries of $\mathbf{D}_{i,j}$,
$\mathbf{F}_2,\ldots,\mathbf{F}_\alpha$ and $\mathbf{c}_2,\ldots,\mathbf{\alpha}$
as constants and the other entries as free variables.
Then each entry of the matrix $\mathbf{M}_1$ in~\eqref{eq:req3} and $\mathbf{M}_2$
in~\eqref{eq:req5} can be interpreted as a polynomial with total degree 2 and 1, respectively.
For the repair condition. The multiplication of all the determinants of the
corresponding sub-matrices in~\eqref{eq:req3} and~\eqref{eq:req5} is a polynomial with total degree
$2\alpha n/r+\alpha t=\alpha(2n/r+t)$.

For the reconstruction condition. Each entry of the matrix in~\eqref{eq:all-encoding-matrices2}
is a polynomial with total degree 1. The multiplication of all the determinants can
be interpreted as a polynomial with total degree
\begin{align*}
\alpha\sum_{i=1}^{\min\{n-k,k\}}i\dbinom{k}{i}\dbinom{n-k}{i}.
\end{align*}
Therefore, the repair condition and MDS property condition are satisfied if the
field size is larger than~\eqref{field-size2} according to the Schwartz-Zippel Lemma.
\end{proof}

\vspace{-7mm}
\begin{IEEEbiographynophoto}{Hanxu Hou} received the B.Eng. degree in Information Security from Xidian University, Xi¡¯an, China, in 2010, and Ph.D. degrees in the Dept. of Information Engineering from The Chinese University of Hong Kong in 2015 and in the School of Electronic and Computer Engineering, Peking University. He is now an Assistant Professor with the School of Electrical Engineering \& Intelligentization, Dongguan University of Technology. His research interests include erasure coding and coding for distributed storage systems.
\end{IEEEbiographynophoto}

\begin{IEEEbiographynophoto}{Patrick P. C. Lee} received the B.Eng. degree (first class honors) in Information
Engineering from the Chinese University of Hong Kong in 2001, the M.Phil.
degree in Computer Science and Engineering from the Chinese University
of Hong Kong in 2003, and the Ph.D. degree in Computer Science from
Columbia University in 2008. He is now an Associate Professor of the
Department of Computer Science and Engineering at the Chinese University
of Hong Kong. His research interests are in various applied/systems topics
including storage systems, distributed systems and networks, operating systems, dependability, and security.
\end{IEEEbiographynophoto}

\begin{IEEEbiographynophoto}{Kenneth W. Shum} received his B.Eng. degree in the Department of Information Engineering from The Chinese University of Hong Kong in 1993, and MSc and Ph.D. degrees in Department of Electrical Engineering from University of Southern California in 1995 and 2000, respectively. He is now an Associate Professor with the School of Science and Engineering, The
Chinese University of Hong Kong (Shenzhen). His research interests include coding for distributed storage systems and sequence design for wireless networks.
\end{IEEEbiographynophoto}

\begin{IEEEbiographynophoto}{Yuchong Hu} received the B.S. degree in Computer Science and Technology from the School of the
Gifted Young, University of Science and Technology of China, Anhui, China, in 2005, and the Ph.D.
degree in Computer Science and Technology from the School of Computer Science, University of Science
and Technology of China, in 2010. He is currently an Associate Professor with the School of Computer Science and Technology, Huazhong University of Science and Technology. His research interests focus on improving the fault tolerance, repair and read/write performance of storage systems, which include cloud
storage systems and key/value stores.
\end{IEEEbiographynophoto}
\end{document}